\documentclass[envcountsame,cleveref]{lipics-v2021}
\usepackage{complexity}
\usepackage{amsmath}

\usepackage{thmtools}
\usepackage{thm-restate}

\usepackage{amssymb}
\usepackage{multirow}
\usepackage{tikz-network}
\usepackage{tikz}
\usetikzlibrary{arrows,automata,positioning,calc}
\usepackage{comment}

\usepackage{hyperref}

\usepackage{pdflscape}



\hideLIPIcs
\nolinenumbers

\usepackage{ifpdf}
\ifpdf
\fi

\newtheorem{claim}{Claim}{\itshape}{\rmfamily}

\usepackage{todonotes}

\usepackage{environ}

\newcommand{\repeattheorem}[1]{%
	\begingroup
	\renewcommand{\thetheorem}{\ref{#1}}%
	\expandafter\expandafter\expandafter\theorem
	\csname reptheorem@#1\endcsname
	\endtheorem
	\endgroup
}

\NewEnviron{reptheorem}[1]{%
	\global\expandafter\xdef\csname reptheorem@#1\endcsname{%
		\unexpanded\expandafter{\BODY}%
	}%
	\expandafter\theorem\BODY\unskip\label{#1}\endtheorem
}

\newcommand{\repeatlemma}[1]{%
	\begingroup
	\renewcommand{\thetheorem}{\ref{#1}}%
	\expandafter\expandafter\expandafter\lemma
	\csname replemma@#1\endcsname
	\endlemma
	\endgroup
}

\NewEnviron{replemma}[1]{%
	\global\expandafter\xdef\csname replemma@#1\endcsname{%
		\unexpanded\expandafter{\BODY}%
	}%
	\expandafter\lemma\BODY\unskip\label{#1}\endlemma
}

\newcommand{\repeatproposition}[1]{%
	\begingroup
	\renewcommand{\thetheorem}{\ref{#1}}%
	\expandafter\expandafter\expandafter\proposition
	\csname repproposition@#1\endcsname
	\endproposition
	\endgroup
}

\NewEnviron{repproposition}[1]{%
	\global\expandafter\xdef\csname repproposition@#1\endcsname{%
		\unexpanded\expandafter{\BODY}%
	}%
	\expandafter\proposition\BODY\unskip\label{#1}\endproposition
}

\newcommand{\repeatcorollary}[1]{%
	\begingroup
	\renewcommand{\thetheorem}{\ref{#1}}%
	\expandafter\expandafter\expandafter\corollary
	\csname repcorollary@#1\endcsname
	\endcorollary
	\endgroup
}

\NewEnviron{repcorollary}[1]{%
	\global\expandafter\xdef\csname repcorollary@#1\endcsname{%
		\unexpanded\expandafter{\BODY}%
	}%
	\expandafter\corollary\BODY\unskip\label{#1}\endcorollary
}

\crefname{lemma}{Lemma}{Lemmata}
\Crefname{lemma}{Lemma}{Lemmata}

\usepackage{etoolbox}

\newif\ifwithappendix
\withappendixfalse

\newcommand{\collectedappendix}{}

\NewEnviron{appendixmaterial}{%
	\ifwithappendix
	\expandafter\gappto
	\expandafter\collectedappendix
	\expandafter{\BODY}%
	\else
	\renewcommand{\repeattheorem}[1]{}%
	\renewcommand{\repeatlemma}[1]{}%
	\renewcommand{\repeatproposition}[1]{}%
	\renewcommand{\repeatcorollary}[1]{}%
	\BODY
	\fi
}

\newcommand{\printcollectedappendix}{%
	\ifwithappendix
	\appendix
	\section{Additional proofs}
	\label{app:collected-proofs}
	\collectedappendix
	\fi
}

\newtheorem{mialc}{Claim}

\newcommand{\resp}{respectively}
\newcommand{\dfa}{DFA}

\newcommand{\N}{\ensuremath{\mathbb{N}}}

\DeclareMathOperator{\sd}{\bf sd}

\renewcommand{\L}{\ComplexityFont{L}}

\renewcommand{\epsilon}{\varepsilon}

\usepackage{graphicx}

\makeatletter
\newcommand{\precneq}{\mathrel{\text{\prec@eq}}}
\newcommand{\prec@eq}{%
	\oalign{%
		\hidewidth$\m@th\prec$\hidewidth\cr
		\noalign{\nointerlineskip\kern1ex}%
		$\m@th\smash{\Neq}$\cr
		\noalign{\nointerlineskip\kern-.5ex}%
	}%
}
\newcommand{\Neq}{\raisebox{0.65ex}{\rotatebox{90}{\scalebox{1}[-1]{$\nshortmid$}}}}
\makeatother

\DeclareMathOperator{\lcm}{lcm}
\DeclareMathOperator{\mpc}{\bf mpc}

\DeclareMathOperator{\mpl}{{\bf mpl}}

\newcommand{\dsc}{\texttt{sc}}

\def\ta{\mathtt{a}}
\def\tb{\mathtt{b}}
\def\tc{\mathtt{c}}

\DeclareMathOperator\ScatFact{ScatFact}

\DeclareMathOperator{\Fix}{Fix}
\newcommand{\Pump}{\operatorname{Pump}}
\newcommand{\NonRep}{\operatorname{NonRep}}

\def\nth#1{#1$^{\text{th}}$}

\newif\ifpaper
\paperfalse 

\newif\iflipics
\lipicstrue 

\definecolor{lime}{HTML}{A6CE39}
\DeclareRobustCommand{\orcidicon}{%
	\begin{tikzpicture}
		\draw[lime, fill=lime] (0,0) 
		circle [radius=0.16] 
		node[white] {{\fontfamily{qag}\selectfont \tiny ID}};
		\draw[white, fill=white] (-0.0625,0.095) 
		circle [radius=0.007];
	\end{tikzpicture}
	\hspace{-2mm}
}

\foreach \x in {A, ..., Z}{%
	\expandafter\xdef\csname orcid\x\endcsname{\noexpand\href{https://orcid.org/\csname orcidauthor\x\endcsname}{\noexpand\orcidicon}}
}

\definecolor{DarkGreen}{RGB}{0,100,0}

\newcommand{\orcidauthorA}{0000-0002-1531-7970}
\newcommand{\orcidauthorB}{0009-0005-1145-5806}
\newcommand{\orcidauthorC}{0000-0003-1008-4642}

\title{Witness Languages for Pumping Lemmata}

\titlerunning{Witness Languages for Pumping Lemmata}

\author{Pamela Fleischmann}
{Kiel University, Germany}
{fpa@informatik.uni-kiel.de}
{\orcidauthorA}{}

\author{Annika Huch}
{Kiel University, Germany}
{ahu@informatik.uni-kiel.de}
{\orcidauthorB}{}

\author{Christian Rauch}
{School of Electrical Engineering and Computer Science, University of Kassel, Germany} 
{christian.rauch@uni-kassel.de}
{\orcidauthorC}{}

\EventEditors{John Q. Open and Joan R. Acces}
\EventNoEds{2}
\EventLongTitle{42nd Conference on Very Important Topics (CVIT 2016)}
\EventShortTitle{CVIT 2016}
\EventAcronym{CVIT}
\EventYear{2016}
\EventDate{December 24--27, 2016}
\EventLocation{Little Whinging, United Kingdom}
\EventLogo{}
\SeriesVolume{42}
\ArticleNo{23}



\Copyright{Pamela Fleischmann, Annika Huch, Christian Rauch} 

\withappendixfalse

\begin{document}
	
	\maketitle
	
	\begin{abstract}
	In 2022 Dassow and Jecker started investigations on the operational complexity of the pumping constant w.r.t.~two pumping lemmata. Since then, the research on pumping lemmata emerged to several variants for regular as well as context-free languages. Within this work, we further investigate the \emph{witnesses} of regular languages (firstly introduced by Rauch) w.r.t.~different pumping lemmata: For a language $L$ and a minimal pumping constant $K$, a \emph{witness} of a language is a word that cannot be \emph{pumped} by $K$ w.r.t.~the given pumping lemma. Here, the pumping constant is minimal if for all smaller numbers, the pumping lemma condition is not satisfied. Alongside combinatorial investigations on (unary) witness languages, we introduce and study the measure of scattered factor distance, which is related to minimal pumping 
	constants.
	We present bounds on minimal pumping constants by the scattered factor distance and show that deciding, given two regular languages, whether one of them is equal to the witness language of the second is \coNP-complete.
	
\end{abstract}

	\newpage
	
	\section{Introduction}
	In parallel to the beginning of automata and formal language theory, researchers developed pumping and iteration properties as analytical tools for classifying formal languages
with the aim of obtaining deeper insights into the computational complexity and expressive power of various types of mechanisms accepting or generating languages.
In fact, the literature on pumping properties
is extensive and covers a broad spectrum of  applications beyond mere language classification; see, for instance, \cite{deLuVa99,EhPaRo81,Ho78}, where language families are defined \textit{via} pumping
properties. To give some examples, in \cite{So80} the focus is on pumping with the additional
constraint that the repetition of a factor, i.e., a contiguous substring of a word, is permitted only if it is done a
minimal number of times. In \cite{Ka14}, the authors examined regular
pumping of Turing machine languages and learnability, while
\cite{MPPS20,PMPPS21} studied the connections between restarting automata, pumping patterns, and formal neural networks.
Recent developments on pumping properties for register and weighted automata on infinite alphabets can be found in \cite{DBLP:conf/stacs/Danieli26,DBLP:conf/stacs/MazowieckiPS26}.
One of the classical pumping lemmata for regular
languages was proven by Kozen in~\cite[page~70, Theorem~11.1]{Ko97}. For an annotated bibliography on pumping lemmata, we refer to~\cite{Ni82}.

\begin{lemma}[Kozen's Pumping Lemma]\label[lemma]{lem:pumping}\label[lemma]{lem:mpc}
	Let~$L\subseteq\Sigma^{\ast}$ be a regular language. Then there is a
	constant~$p$ (depending on~$L$) such that for all
	$w\in L$ with $|w| \geq p$ there exist $x,y,z\in \Sigma^*$,\
	such that $w = xyz$, $|y|>0$, and
	$xy^tz \in L$ for all~$t\geq 0$. We say that $y$ can be
	\emph{pumped} in~$w$.
\end{lemma}

Given a regular language~$L$, the value of~$p$ in the aforementioned
pumping lemma can always be chosen to be the number of states of some
finite automaton, regardless of whether it is deterministic or
non-deterministic, that accepts~$L$. 
Consider the regular language
$L_1 = \{(\ta \tb)^n \ta \ta \mid n \geq 1 \}\cup\{\ta\}$ (see \Cref{fig:automaton-abnaa} for a DFA accepting~$L$):
\begin{figure}
	\centering
	\begin{tikzpicture}[->,>=stealth',shorten >=1pt,auto,node distance=1.1cm,
		semithick,scale=0.7, every node/.style={scale=0.7},initial text=]
		\tikzstyle{every state}=[fill={rgb:black,1;white,10},ellipse]
		
		\node[state,initial] (q0) [thick] {$q_0$};	
		\node[state,right=of q0,accepting] (q1) [thick] {$q_1$};
		\node[state,right=of q1] (q2) [thick] {$q_2$};	
		\node[state,right=of q2] (q3) [thick] {$q_3$};	
		\node[state,accepting,right=of q3] (q4) [thick] {$q_4$};	
		\node[state,right=of q4] (q5) [thick] {$q_5$};	
		
		\path [-stealth, thick]
		(q0) edge node {$\ta$}   (q1)
		(q0) edge[bend right] node {$\tb$}   (q5)
		(q1) edge  node {$\tb$} (q2)
		(q1) edge[bend right]  node {$\ta$}(q5)
		(q2) edge node  {$\ta$} (q3)
		(q2) edge[bend right]  node {$\tb$}(q5)
		(q3) edge  node {$\ta$} (q4)
		(q3) edge[bend right] node[above] {$\tb$}(q2)
		(q4) edge  node {$\ta,\tb$} (q5)
		(q5) edge [loop right]  node {$\ta,\tb$}(q5)
;
	\end{tikzpicture}

	\caption{A \dfa\ accepting $\{(\ta \tb)^n \ta \ta \mid n \geq 1\}\cup\{\ta\}$.}
	\label{fig:automaton-abnaa}
\end{figure}
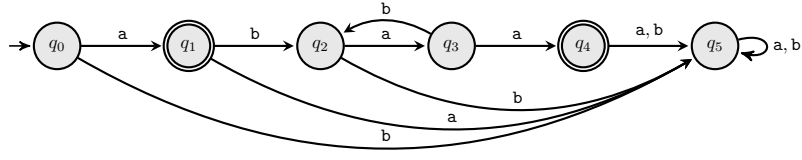
each word in~$L$ that contains the factor~$(\ta \tb)^2$ can be pumped by the factor~$\ta \tb$. 
However, the word~$\mathtt{abaa}$ cannot be pumped: the lemma forces us to also pump by $t=0$, i.e., deleting a factor, and the only shorter word in~$L$ is~$\ta$ but $(\mathtt{aba})^2\ta$ as well as~$\ta(\mathtt{baa})^2$ are not in~$L$ which implies that~$\mathtt{abaa}$ cannot be pumped by any of its factors.
Consequently, all values $0\leq p\le 4$ do
\emph{not} satisfy the property of the pumping lemma, but $p=5$
does.
A closer look at some exemplary languages reveals that sometimes a
much smaller value for $p$ than the number of states suffices. For instance, consider the language~$L_2=\ta^* + \ta^*\tb\tb^*$,
given as a regular expression.
Indeed, the language~$L_2$ is accepted by a (minimal) DFA with three states, the sink state included.
Nevertheless, the statement of the pumping lemma is already satisfied for $p=1$ which is minimal since $p=0$ is equivalent to $L=\emptyset$~\cite{DaJe22}.

This motivates the notion of a \emph{minimal pumping constant} for a language~$L$ with respect to 
a particular pumping lemma, which is the smallest number~$p$ such that
the pumping lemma under consideration for the language~$L$ is
satisfied. 
Formally, the smallest pumping constant~$p$ satisfying the above pumping lemma
for a (not necessarily regular) language~$L$ is denoted by~$\mpc(L)$, where we define~$\mpc(L)=\infty$
if there is no natural number~$p$ such that the pumping lemma is fulfilled.
In recent years, these minimal pumping constants have attracted attention. Starting with an investigation of the operational complexity of the aforementioned pumping constants~\cite{DaJe22}, subsequent work has focused on their descriptional complexity~\cite{HoRa23a,HoRa23d}.
In addition, the computational complexity of the associated minimisation problem (\textsc{Pumping-Problem}), which asks for the smallest value~$p$ satisfying a given pumping lemma, has been analysed for both regular~\cite{GHR23} and non-regular language families~\cite{GrHoRa24b,HoRa25}.
Most recently, in a broader area of this topic, the relation between the size of a given grammar and its minimal pumping constant has been studied in~\cite{GrHo25,MaMaZe26}.

To demonstrate that a language does \emph{not} exhibit a particular pumping property, one typically identifies a word that cannot be pumped with respect to any value smaller than the minimal pumping constant.
We refer to such a word as a \emph{witness} of the language (a notion introduced in \cite{Ra25}) and to the set of all witnesses of a given regular language as its \emph{witness language}.
By the prior example, every non-empty word in $L_2$ can be pumped by its first letter; thus the only word in the witness language of~$L_2$ is~$\lambda$ which cannot be pumped.
Note that the words contained in a witness language may be shorter than the minimal pumping constant, e.g. $L_1$ (see \Cref{fig:automaton-abnaa} for a DFA accepting~$L_1$).
As already shown, every word in~$L_1$ containing at least~$5$ letters can be pumped, whereas~$\mathtt{abaa}$ cannot be pumped. Moreover, the word~$\ta\in L_1$ cannot be pumped because it is the shortest word in~$L_1$.
Consequently, the witness language of~$L_1$ is given by~$\{\ta, \mathtt{abaa}\}$.

%
%

\textbf{Own Contribution.}
We start with an investigation on the smallest powers that can be used for a non-regular language to show its non-regularity using one of the pumping lemmata. Here, we relate the fields of pumping lemmata to the one of scattered factors by the scattered factor distance.
Afterwards, we move the focus back to regular languages by investigating the set of words of a regular language that are not pumpable by some fixed power $t$. Here, we start with general results on witness languages and then move on to unary regular languages. Moreover, we investigate the \textsc{Pumping-Witness-Language-Problem} which is to decide whether a regular language is the witness language of a second given regular language and show the problem's \coNP-completeness.

\textbf{Structure of the Work.}
After recalling the relevant Pumping Lemmata and fixing the necessary notation for the study of witness languages, we first inspect a witness property for non-regular languages and their connection to the scattered-factor distance within a language in \Cref{sec:smallest-pumping-non-reg}. In \Cref{sec:non-unary-witness} we present general results on witness languages of regular languages. Since the situation for general regular languages gets more involved, \Cref{sec:unary-witness} narrows the investigation down to unary regular languages. Lastly, \Cref{sec:witness-problem} is dedicated to the Pumping-Witness-Language-Problem which is to decide for two languages and a pumping lemma whether one language is the witness language of the second one. 

Further, \Cref{sec:witness-language-operator} contains an AI-assisted study of the witness language operator. Note that the usage of AI is restricted to \Cref{sec:witness-language-operator} only.


	\section{Preliminaries}
	Let $\mathbb{N}$ be the set of all natural numbers, $\mathbb{N}_0 = \mathbb{N} \cup \{0\}$, $[n] = \{1,\ldots,n\}$, $[i,j]=[j]\backslash[i-1]$, $[n]_0 = [n] \cup \{0\}$, and $[n]_{>i} = \{i+1,\ldots, n\}$ for some $i\leq j <n$. 
An \emph{alphabet} $\Sigma$ is a non empty, finite set whose elements are called \emph{letters} or {\em symbols}.
A \emph{word} is a finite sequence of letters from $\Sigma$. Let $\Sigma^*$ be 
the set of all finite words over $\Sigma$ with concatenation and the empty word $\lambda$ as 
neutral element.
Set $\Sigma^+ = \Sigma^* \setminus \{\lambda\}$. Let $w\in\Sigma^{\ast}$. 
For all $n\in\N_0$, inductively define $w^0=\lambda$ and $w^n=ww^{n-1}$.
The {\em length} $|w|$ of $w$ is the number of $w$'s letters; thus $|\lambda| = 0$. 
For all $k \in \mathbb{N}_0$ set $\Sigma^k = \{w \in \Sigma^* \mid |w| = k\}$ (define $\Sigma^{\leq k}$, $\Sigma^{\geq k}$ analogously). We denote $w$'s \nth{$i$} letter by $w[i]$ 
and set $w[i..j]=w[i]\cdots w[j]$ if $i < j$,  and $\lambda$ if $ i > j$ for all $i,j \in [\vert w \vert]$. 
For each $\ta \in \Sigma$ set $|w|_{\ta} = |\{i \in [|w|] \mid w[i]=\ta \}|$. The word $u \in 
\Sigma^*$ is called a \emph{factor} of $w$ if there exist $x,y \in \Sigma^*$ such that $w = xuy$.
In the case $x=\lambda$, we call $u$ a \emph{prefix} of $w$ and \emph{suffix} if $y = \lambda$.
Further, for $w = xy \in \Sigma^*$ we define $x^{-1}w = y$ and $w y^{-1} = x$. 
By $w^R$ we denote the reversal of $w$, i.e., $w^R = w[|w|] \cdots w[1]$.
A word $u\in\Sigma^{\ast}$ is called a {\em scattered factor} of a word $w\in\Sigma^{\ast}$
if there exist possibly empty words $v_1,\ldots,v_{|u|+1}\in\Sigma^{\ast}$ such that $w=v_1u[1]v_2u[2]\cdots v_{|u|}u[|u|]v_{|u|+1}$. We denote the set of scattered factors of a word $w$ by $\ScatFact(w)$.
A {\em (formal) language} $L$ is a subset of $\Sigma^{\ast}$.
A language~$L\subseteq\Sigma^*$ is \emph{suffix closed} if $L=\{\,x \mid \exists y\in\Sigma^{\ast}:\, yx\in L\}$.
We denote by $\mathcal{FIN}_\Sigma$ the set of all finite languages over $\Sigma$ (if $\Sigma$ is known from the context, the subscript is omitted). For abbreviation, we set $L_{\geq p}=L\cap\Sigma^{\geq p}$ for $L\subseteq \Sigma^{\ast}$ and $p\in\N$. Further, $L^R = \{w^R \mid w \in L\}$.


Next we fix some definitions on finite automata (cf.~\cite{Ha78}). 
A \emph{deterministic finite automaton} (\dfa)
is a quintuple $A=(Q,\Sigma,{}\cdot{},s,F)$, where~$Q$ is the finite set
of \textit{states}, $\Sigma$ is the finite set of \textit{input
	symbols}, $s\in Q$ is the \textit{initial state}, $F\subseteq Q$ is
the set of \textit{accepting states}, and the \textit{transition function}~$\cdot$ maps~$Q\times\Sigma$ to~$Q$; written infix and intuitively recursively  extended to $\Sigma^{\ast}$.
The \emph{language accepted}
by the \dfa~$A$ is defined as $L(A) =\{\,w\in \Sigma^*\mid\,s\cdot w\in F\}$.
A \dfa\ is \textit{unary}, 
if $|\Sigma|=1$.
The \emph{state complexity} of a language~$L$ accepted by a \dfa\ is denoted by
$$\dsc(L)=\min\{\,\dsc(A)\mid\mbox{$A$ is a \dfa\ with $L=L(A)$}\,\},$$
where $\dsc(A)$ refers to the number of states of the
automaton~$A$.
%
%
A finite automaton is \emph{minimal} if its number of states is
minimal w.r.t.~the accepted language. 
The syntactic monoid for a given language $L\subseteq \Sigma^*$, is
defined by the \emph{syntactic congruence}~$\equiv_L$ over~$\Sigma^*$
where $v_1\equiv_L v_2$ iff $uv_1w\in L\iff uv_2w\in L$ for
every $u,w\in\Sigma^*$.  Then the \emph{syntactic monoid} is the
quotient monoid $M(L)=\Sigma^*/\!\equiv_L$ where the concatenation of
equivalence classes $[u]_{\equiv_L}\cdot [v]_{\equiv_L}=[uv]_{\equiv_L}$
serves as the monoid operation. 
The syntactic monoid of a regular language~$L$ is the smallest monoid
recognising the language under consideration (w.r.t.~the
division relation) and it is isomorphic to the transformation monoid
of the minimal deterministic finite automaton
accepting~$L$. 
A language $L\subseteq\Sigma^*$ is \emph{recognisable} if and only if
there exists a finite monoid~$M$, a morphism
$\varphi_L:\Sigma^*\rightarrow M$, and a subset $N\subseteq M$ such
that $L=\varphi_L^{-1}(N)$, which in turn is equivalent to the
regularity (acceptance by a finite state machine) of~$L$.
For an $n$-state \dfa\ the size of the
transformation monoid is at most~$n^n$.


Regular languages satisfy a variety of different pumping lemmata---for
a comprehensive list of pumping or iteration lemmata we refer
to~\cite{Ni82}.
%

The next lemma is a stronger variant of Lemma~\ref{lem:mpc}, which can be found in e.g.~\cite[p.~119, Lemma~8]{RaSc59},
\cite[p.~252, Folgerung~5.4.10]{Br84}, and~\cite[p.~56,
Lemma~6.1]{HoUl79}, and reads as follows:

\begin{lemma}\label[lemma]{lem:pumping-stronger-version}\label[lemma]{lem:mpl}
	Let~$L\subseteq\Sigma^{\ast}$ be a regular language. Then there is a
	constant~$p \in \N_0$ (depending on~$L$) such that for all
	$w\in L$ with $|w| \geq p$ there exist $x,y,z\in \Sigma^*$,\
	such that $w = xyz$, $|y|>0$, $|xy|\leq p$ and
	$xy^tz \in L$ for all~$t\geq 0$.
	We again say that~$y$ can be \emph{pumped} in~$w$.
\end{lemma}

Note that Lemma~\ref{lem:pumping-stronger-version} implies Kozen's pumping
lemma as the condition $|xy|\leq p$ is \emph{not} required. 

%

\begin{definition}
The smallest pumping constant~$p$ satisfying Kozen's pumping lemma
for a (not necessarily regular) language~$L$ is denoted by~$\mpc(L)$.
Analogously, for each language~$L$ let~$\mpl(L)$ denote the smallest value~$p$ that satisfies Lemma~\ref{lem:mpl}.
Set $\mpc(L)$ and $\mpl(L)$ resp. to $\infty$, if there is no natural number~$p$ such that the resp. pumping lemma is fulfilled.
\end{definition}

\begin{remark}
Pumping lemmata were considered in~\cite{DaJe22}, where 
besides some immediate facts such as \\
- $\mpc(L)=0$ iff $\mpl(L)=0$ iff $L=\emptyset$,\\
- $\mpc(L)=\mpl(L)= 1+\max\{\,|w|\mid w\in L\,\}$ for finite language $L\neq\emptyset$,\\
- $\mpc(L)=1$ implies $\lambda\in L$, and\\
- if $\mpl(L)=1$, then~$L$ is suffix closed,\\
also the inequalities $\mpc(L)\leq\mpl(L)\leq\dsc(L)$
and results on the operational complexity w.r.t.\ these minimal
pumping constants were shown.
\end{remark}

\begin{definition}
Let~$L \subseteq\Sigma^*$ and $K\in\{\mpc,\mpl\}$.
For a word~$w\in L$ we say that~$w$ can be \emph{~$K$-pumped} w.r.t. the value~$p \in \N$ and the language $L$ if there is a decomposition of the word~$w$ that fulfils the conditions of the pumping lemma associated to~$K$ w.r.t~$p$ and $L$. Then~$w$ is \emph{~$K$-pumpable}.
If the pumping lemma (the language) is clear from the context, we omit the prefix~$K$ (the additional phrase w.r.t.~$L$).
\end{definition}

The following definition captures the set of our interest containing witnesses for being not less than $K$-pumpable, i.e., the {\em witness language}. Note, that this definition is a more relaxed version of the definition of witnesses in~\cite{Ra25}. Indeed, there are pumping lemmata such that the pumping property also takes the complement of~$L$ into account, e.g., see~\cite{Ja78}. Then we could simply allow the word~$w$ to behave complementary as well. Since we are only dealing with pumping properties regarding the language-and not its complement-the definition is restricted to the case~$w\in L$.

\begin{definition}
Let~$L$ be a language over the alphabet~$\Sigma$ and let $K\in\{\mpc,\mpl\}$.
The word~$w\in L$ is called a \emph{$K$-witness} if there is no $p < K(L)$ such that $w$ is $K$-pumpable w.r.t. $p,L$. If $K$ is clear from the context, we call $w$ just a \emph{witness}.
For a given regular language~$L$, let $W^{K}(L)$ denote the \emph{witness language} defined as the set of $K$-witnesses of~$L$. 
If $K$ is clear from the context, we write~$W(L)$ instead.
\end{definition}

It is clear that the witness language depends on the minimal pumping constant of the underlying language. Hence, we want to examine which witness languages can be obtained for a fixed minimal pumping constant.

\begin{definition}
Let $K\in\{\mpc,\mpl\}$, and~$p \in \N$. Then we set
$$\mathcal{W}_p^K:=\{\,W^K(L)\mid L\subseteq \Sigma^* \mbox{is a regular language with } K(L)=p\}$$
and
$$\mathcal{W}_p^{K,u}:=\{\,W^K(L)\mid L\subseteq \Sigma^* \mbox{is a unary regular language with } K(L)=p\}$$
to be the sets of (unary) witness languages for the value~$p$.
\end{definition}

Further, we introduce the measure {\em scattered factor distance} that will be compared later to the pumping complexities.
Let
\(
\overline{\N}_0=\N_0\cup\{\infty\}
\)
 be equipped with its natural order.
For a language
$L\subseteq\Sigma^*$, we define its \emph{scattered-factor
	distance} by
\[
\sd(L)
=
\sup_{w\in L}
\min
\left\{
|w|-|w'|
\;\middle|\;
w'\in
\{\lambda\}
\cup
\Bigl(
(\ScatFact(w)\setminus\{w\})\cap L
\Bigr)
\right\}.
\]
The supremum is taken in~$\overline{\N}_0$; in particular, we
use the convention
\(
\sup\emptyset=0.
\)
%
Note that we explicitly allow~$\sd(L)$ to be equal to~$\infty$.
A language~$L$ is called \emph{(scattered factor) dense} if~$\sd(L)<\infty$.

	\section{The Witness Property within non-regular Languages}
	\label{sec:smallest-pumping-non-reg}
	To prove that a given language~$L$ is not contained in a certain language class~$\mathcal{L}$ one can try to find a witness~$w \in L$ which contradicts a pumping property of~$\mathcal{L}$.
In terms of Lemma~\ref{lem:mpc} one has to show that for each decomposition~$w=xyz$ with~$y\ne\lambda$ there is an integer~$t\ge 0$ such that~$xy^tz\notin L$.
We set $\mathcal{L}_{\mpc} \subset 2^{\Sigma^*} \setminus REG$ ($\mathcal{L}_{\mpl} \subset 2^{\Sigma^*} \setminus REG$, \resp) to be the set of non-regular languages that do not satisfy the Pumping Lemma~\ref{lem:mpc} (\ref{lem:mpl}, \resp), where~$2^{\Sigma^*}$ denotes the power set of~$\Sigma^*$.
In this section, we will study the languages which can be proven to be in~$\mathcal{L}_{\mpc}$ or in~$\mathcal{L}_{\mpl}$ by using a given value~$t$. 
Thus, given a $t\in\N$, we investigate  $L \in \mathcal{L}_{\mpc}$ (or $L\in \mathcal{L}_{\mpl}$, \resp) that fulfil the following property: for all constants $p$, there exists $w \in {L}_{\geq p}$ such that for all $x,y,z \in \Sigma^*$ with $w= xyz$ we have $xy^t z \notin L$. Note that this immediately implies for $L \in \mathcal{L}_{K}$ that $K(L) = \infty$ for $K \in \{\mpc,\mpl\}$. The connection we establish here to the scattered factor distance is useful in the later sections.

\begin{definition}\label[definition]{deftpowerclose}
	A language $L$ is called \emph{$t$-power closed} for $t \in \N_0$ if for all $w \in L$ and any factorisation $w = xyz$ and $x,y,z \in \Sigma^*$ with $|y| \geq 1$ we have $xy^tz \in L$.
\end{definition}

\begin{remark}\label[remark]{higman}
	Every $0$-power closed language is a scattered factor closed language. By the famous result of Higman \cite{Hi52} we can conclude that this language is regular.
\end{remark}

Next, we establish the following subset relation between~$\mathcal{L}_{\mpc}$ and $\mathcal{L}_{\mpl}$ to obtain the results 
that directly hold for both sets.

\ifpaper
\begin{proposition}\label{lem:Lmpc-subsetneq-Lmpl}
	We have~$\mathcal{L}_{\mpc}\subsetneq \mathcal{L}_{\mpl}$.
\end{proposition}
\begin{proof}
	Since Lemma~\ref{lem:mpl} implies Lemma~\ref{lem:mpc} we have~$\mpc(L)\le \mpl(L)$
	for all languages~$L$, i.e., $\mpc(L)=\infty$ directly implies~$\mpl(L)=\infty$ and thus~$\mathcal{L}_{\mpc}\subseteq \mathcal{L}_{\mpl}$. 
	
	For showing that the inclusion is proper let~$\Sigma=\{\ta,\tb\}$. We set
	$$PAL=\{ww^R,w \tc w^R\mid w\in \Sigma^*, \tc\in\Sigma\}$$ to be the language of palindromes over~$\Sigma$.
	Observe that each word~$w\in PAL$ can be pumped by~$w$ since~$w^k\in PAL$ for each~$k\ge 0$.
	Hence, we have~$\mpc(PAL)=1$.
	On the other hand, assume that~$\mpl(PAL)=p<\infty$ for a constant~$p$.
	We observe that the word~$w=uu^R$ with~$u=\ta\tb\ta\tb^2\ta\tb^3\ta\cdots \tb^{3p-1}\ta$
	cannot be pumped by any factor of its prefix of length~$p$ which can be seen as follows.
	Assume we have~$w=xyz$ with~$|xy|\le p$ and~$y\ne\lambda$. Since~$|u|=1+\sum_{i=1}^{3p-1}(i+1)=(3p+1)3p/2>p$ the words~$xy$ cannot contain letters from~$u$ and~$u^R$.
	Therefore, we have that the suffix of~$xy^3z$ is~$u^R$ while~$u$ is not the prefix of~$xy^{p+1}z$ anymore which can be seen as follows.
	First observe that each infix of~$u$ containing at least two $\ta$'s appears exactly once in~$u$. Indeed, if~$y$ would contain the letter~$\ta$ at least once the word~$xy^3$ contains the factor~$yy$ twice which contains the letter~$\ta$ twice.
	Hence, $xy^3$ cannot be a prefix of~$u$ and thus~$u$ is not a prefix of~$xy^{p+1}z$.
	On the other hand, if~$y$ contains no $\ta$ then we have~$xy^2=\ta\tb\ta\tb\tb\ta\cdots\ta\tb^j\ta\tb^{j+1-|y|+2|y|}$ which is not a prefix of~$u$ and thus again~$u$ is not a prefix of~$xy^{p+1}z$.
	Therefore, the word~$w$ cannot be pumped w.r.t. $p$ and we have~$\mpl(PAL)>p$ which is a contradiction.\iflipics\else\qed\fi
\end{proof}
\else
\begin{proposition}\label{lem:Lmpc-subsetneq-Lmpl}
	We have~$\mathcal{L}_{\mpc}\subsetneq \mathcal{L}_{\mpl}$.
\end{proposition}
\begin{proof}
	Since Lemma~\ref{lem:mpl} implies Lemma~\ref{lem:mpc} we have~$\mpc(L)\le \mpl(L)$
	for all languages~$L$, i.e., $\mpc(L)=\infty$ directly implies~$\mpl(L)=\infty$ and thus~$\mathcal{L}_{\mpc}\subseteq \mathcal{L}_{\mpl}$. 
	
	For showing that the inclusion is proper let~$\Sigma=\{\ta,\tb\}$. We set
	$$PAL=\{ww^R,w \tc w^R\mid w\in \Sigma^*, \tc\in\Sigma\}$$ to be the language of palindromes over~$\Sigma$.
	Observe that each word~$w\in PAL$ can be pumped by~$w$ since~$w^k\in PAL$ for each~$k\ge 0$.
	Hence, we have~$\mpc(PAL)=1$.
	On the other hand, assume that~$\mpl(PAL)=p<\infty$ for a constant~$p$.
	We observe that the word~$w=uu^R$ with~$u=\ta\tb\ta\tb^2\ta\tb^3\ta\cdots \tb^{3p-1}\ta$
	cannot be pumped by any factor of its prefix of length~$p$ which can be seen as follows.
	Assume we have~$w=xyz$ with~$|xy|\le p$ and~$y\ne\lambda$. Since~$|u|=1+\sum_{i=1}^{3p-1}(i+1)=(3p+1)3p/2>p$ the words~$xy$ cannot contain letters from~$u$ and~$u^R$.
	Therefore, we have that the suffix of~$xy^3z$ is~$u^R$ while~$u$ is not the prefix of~$xy^{p+1}z$ anymore which can be seen as follows.
	First observe that each infix of~$u$ containing at least two $\ta$'s appears exactly once in~$u$. Indeed, if~$y$ would contain the letter~$\ta$ at least once the word~$xy^3$ contains the factor~$yy$ twice which contains the letter~$\ta$ twice.
	Hence, $xy^3$ cannot be a prefix of~$u$ and thus~$u$ is not a prefix of~$xy^{p+1}z$.
	On the other hand, if~$y$ contains no $\ta$ then we have~$xy^2=\ta\tb\ta\tb\tb\ta\cdots\ta\tb^j\ta\tb^{j+1-|y|+2|y|}$ which is not a prefix of~$u$ and thus again~$u$ is not a prefix of~$xy^{p+1}z$.
	Therefore, the word~$w$ cannot be pumped w.r.t. $p$ and we have~$\mpl(PAL)>p$ which is a contradiction.\iflipics\else\qed\fi
\end{proof}
\fi

As a consequence, we will state our claims just for~$\mathcal{L}_{\mpl}$.

\begin{remark}
	Let $L \in \mathcal{L}_{\mpl}$ with the variables presented above. Directly by Definition~\ref{deftpowerclose} we have that if  $x y^tz \notin L$ for some $t \in \N_0$ then the language~$L$ is not $t$-power closed.
\end{remark}

Next we want to study the non-regular languages~$L$ which have this property.

\begin{definition}
	Let~$L\subseteq \Sigma^*$  and~$t\in \N_0$. The language~$L$ is \emph{$t$-witnessable} w.r.t. Lemma~\ref{lem:mpc} (Lemma~\ref{lem:mpl}, resp.) if
	for each $p\in\N$ there exists $w \in L_{\ge p}$ such that for each decomposition~$w=xyz$ with~$y\ne \lambda$ (and~$|xy|\le p$, resp.) we have~$xy^tz\notin L$.
\end{definition}

Roughly speaking, the exponent~$t$ suffices to show that the language~$L$ is not regular.
We start by studying a non-regular language that is not $0$-witnessable.

%
%

\begin{example}\label[example]{exa:mnl-language}
	Let~$L=\{\ta^m \tb^n \tc^\ell\mid 0\le m\le n\le \ell\}$.
	For the sake of contradiction,
	assume that~$L$ is regular and thus~$\mpl(L)=p<\infty$ for $p \in \N$.
	Since~$L$ is not the empty language we know that~$p>0$ due to the basic facts from the preliminaries.
	Then the word~$w=\ta^p \tb^p \tc^p$ fulfilling~$|w|=3p>p$ can be pumped such that~$w=xyz$ with~$|xy|\le p$ and~$|y|>1$.
	Since~$|\tb^p|=p$, the factor~$y$ cannot contain both an~$\ta$ and a~$\tc$.
	If~$y$ contains no~$\tc$ then we have~$|xy^2z|_{\ta}> |xy^2z|_{\tc}$ or $|xy^2z|_{\tb}> |xy^2z|_{\tc}$ which implies that~$xy^2z$ is not in~$L$.
	On the other hand, the word~$xy^0z$ fulfils~$|xy^0z|_{\ta}> |xy^0z|_{\tb}$ if~$y$ contains a $\tb$ and no~$\ta$ or~$|xy^0z|_{\tb}> |xy^0z|_{\tc}$ if~$y$ contains neither a~$\tb$ nor an~$\ta$.
	Hence, we conclude that~$\mpl(L)=\infty$ and thus~$L\in\mathcal{L}_{\mpl}$.
	
	On the other hand, to show non $0$-witnessability of $L$ we need to show that there exists a $p \in \N$ such that for all $w \in L_{\geq p}$ there exists a decomposition $w = xyz$ with $x \neq \lambda$ and $|xy| \leq p$ we have $xy^0 z \in L$. Consider some $p \in \N$. Every word in $L_{\geq p}$ is of one of the following forms: $\ta^{\ell_1}\tb^{\ell_2}\tc^{\ell_3}$ or $\tb^{\ell_2}\tc^{\ell_3}$ or $\tc^{\ell_3}$ for $0 \leq \ell_1 \leq \ell_2 \leq \ell_3$ and $\ell_1 + \ell_2 + \ell_3 \geq p$. For the word~$w=\ta^{\ell_1}\tb^{\ell_2}\tc^{\ell_3}$ with~$\ell_1\ge 1$ we observe for the decomposition~$w=xyz$ with~$x=\lambda$, $y=\ta$, and~$z=y^{-1}w$ the word~$xy^0z=\ta^{\ell_1-1}\tb^{\ell_2}\tc^{\ell_3}$ is in~$L$. For the word~$w=\tb^{\ell_2}\tc^{\ell_3}$ with~$\ell_2\ge 1, \ell_2 + \ell_3 \geq p$ we observe for the decomposition~$w=xyz$ with~$x=\lambda$, $y=\tb$, and~$z=y^{-1}w$ the word~$xy^0z=\tb^{\ell_2-1}\tc^{\ell_3}$ is in~$L$. And last, for the word~$w=\tc^{\ell_3}$ with~$\ell_3\ge p$ we observe for the decomposition~$w=xyz$ with~$x=\lambda$, $y=\tc$, and~$z=y^{-1}w$ the word~$xy^0z=\tc^{\ell_3-1}$ is in~$L$. Hence, the language~$L$ is \emph{not} $0$-witnessable w.r.t. Lemma~\ref{lem:mpl} and thus it is \emph{not} $0$-witnessable w.r.t. Lemma~\ref{lem:mpc}.
\end{example}

Next we observe that the property of being $0$-witnessable is related to the property of not being (scattered-factor) dense.

\ifpaper

\begin{repproposition}{lem:sd-infty-implies-zero-witness}
	Let~$L\in\mathcal{L}_{\mpl}$. If~$\sd(L)=\infty$ then~$L$ is~$0$-witnessable.
\end{repproposition}
\begin{proof}
	Let~$\sd(L)=\infty$
	and $p\in \N$. We want to show that~$L$ contains a word~$w$ whose factorisations~$w=xyz$ with $|y| \geq 1$ and $|xy| \leq p$ fulfil $xy^0 z \notin L$.
	Since~$\sd(L)=\infty$ we obtain that there is a word~$w\in L$
	such that the longest scattered factor~$w'\in L$ of~$w$ has length at most equal to~$|w|-p-1$, i.e., 
	\[(\ScatFact(w)\cap L)\cap (\Sigma^{\ge |w|-p}\cap \Sigma^{\le |w|-1})=\emptyset.
	\]
	We observe that each decomposition~$w=xyz$ with~$|xy|\le p-1$ and~$|y|\ge 1$  fulfils~$|xz|\ge |xyz|-p\ge|w|-p$.
	Since~$xy^0z=xz$ is a scattered factor of~$w=xyz$ we thus obtain that~$xy^0z$ is not in~$L$.
	\iflipics\else\qed\fi
\end{proof}

\else

\begin{repproposition}{lem:sd-infty-implies-zero-witness}
	Let~$L\in\mathcal{L}_{\mpl}$. If~$\sd(L)=\infty$ then~$L$ is~$0$-witnessable.
\end{repproposition}
\begin{proof}
	Let~$\sd(L)=\infty$
	and $p\in \N$. We want to show that~$L$ contains a word~$w$ whose factorisations~$w=xyz$ with $|y| \geq 1$ and $|xy| \leq p$ fulfil $xy^0 z \notin L$.
	Since~$\sd(L)=\infty$ we obtain that there is a word~$w\in L$
	such that the longest scattered factor~$w'\in L$ of~$w$ has length at most equal to~$|w|-p-1$, i.e., 
	\[(\ScatFact(w)\cap L)\cap (\Sigma^{\ge |w|-p}\cap \Sigma^{\le |w|-1})=\emptyset.
	\]
	We observe that each decomposition~$w=xyz$ with~$|xy|\le p-1$ and~$|y|\ge 1$  fulfils~$|xz|\ge |xyz|-p\ge|w|-p$.
	Since~$xy^0z=xz$ is a scattered factor of~$w=xyz$ we thus obtain that~$xy^0z$ is not in~$L$.
	\iflipics\else\qed\fi
\end{proof}

\fi

The following lemma is an immediate consequence of the definition of~$\sd$ and will be helpful to show that the converse statement of Lemma~\ref{lem:sd-infty-implies-zero-witness} does not hold.

\begin{lemma}\label[lemma]{lem:sd-reversed}
	For a given language~$L$  we have~$\sd(L)=\sd(L^R)$.
\end{lemma}

\begin{example}
	Indeed, by reinspecting Example~\ref{exa:mnl-language}, we observe that~$L$ contains the suffix of length~$|w|-1$ for each word~$w\in L$.
	Hence, we have~$\sd(L)=1$. 
	Therefore, we have that~$\sd(L^R)=\sd(\{\tc^\ell \tb^n \ta^m\mid 0\le m\le n\le \ell\})=1$ due to Lemma~\ref{lem:sd-reversed}.
	On the other hand, we have that the word~$w=\tc^p \tb^p \ta^p$ fulfils that 
	$|w| \geq p$ and for every factorisation $w =xyz$ with $|y| \geq 1$ and $|xy| \leq p$, we have $xy^0 z \notin L$ because~$y\in \ScatFact(\tc^p)$ and~$|xy^0z|_{\tc}\le p-1<p=|xy^0z|_{\tb}$. 
	This implies that~$L$ is $0$-witnessable w.r.t.~\Cref{lem:mpl}, but not to \Cref{lem:mpc}. Hence, the converse of \Cref{lem:sd-infty-implies-zero-witness} does not hold.
\end{example}

Now we show that the measure~$\sd$ gives a general lower bound for the pumping constant~$\mpl$.

\ifpaper
\begin{theorem}\label{lem:sd-lower-bound-for-mpl}
	For every language~$L$ we have~$\sd(L)\le \mpl(L)$.
\end{theorem}
\begin{proof}
	Since the statement is trivially fulfilled for~$\mpl(L)=\infty$ we assume for the rest of the proof that~$\mpl(L)<\infty$.
	For the sake of contradiction assume that~$\mpl(L)<\sd(L)$.
	Given two words~$w,w'\in L$ such that~$w' \in \ScatFact(w)$, $\sd(L)\ge |w|-|w'|>\mpl(L)$ and there is no word~$\tilde{w}\in L\setminus\{w,w'\}$ with~$w' \in \ScatFact(\tilde{w})$ and $\tilde{w}\in \ScatFact(w)$.
	Observe that~$|w|\ge |w|-|w'|>\mpl(L)$ which implies that~$w$ is $\mpl$-pumpable.
	Hence, there are words~$x,y,z\in \Sigma^*$ with~$w=xyz$, $y\ne \lambda$, $|xy|\le \mpl(L)$, and~$xy^iz\in L$ for each $i \in \N$.
	In particular, we have~$xy^0z\in L$. On the other hand, we see that~$|xy^0z|=|xz|<|w|$. Since~$w'$ is the longest scattered factor of~$w$ in~$L$ we derive that~$|xz|\le |w'|$. 
	Therefore, we obtain that~$|y|=|w|-|xz|\ge |w|-|w'|>\mpl(L)$ which contradicts the assumption~$|y|\le \mpl(L)$ and therefore proves the stated claim.
	\iflipics\else\qed\fi
\end{proof}
\else
\begin{theorem}\label{lem:sd-lower-bound-for-mpl}
	For every language~$L$ we have~$\sd(L)\le \mpl(L)$.
\end{theorem}
\begin{proof}
	Since the statement is trivially fulfilled for~$\mpl(L)=\infty$ we assume for the rest of the proof that~$\mpl(L)<\infty$.
	For the sake of contradiction assume that~$\mpl(L)<\sd(L)$.
	Given two words~$w,w'\in L$ such that~$w' \in \ScatFact(w)$, $\sd(L)\ge |w|-|w'|>\mpl(L)$ and there is no word~$\tilde{w}\in L\setminus\{w,w'\}$ with~$w' \in \ScatFact(\tilde{w})$ and $\tilde{w}\in \ScatFact(w)$.
	Observe that~$|w|\ge |w|-|w'|>\mpl(L)$ which implies that~$w$ is $\mpl$-pumpable.
	Hence, there are words~$x,y,z\in \Sigma^*$ with~$w=xyz$, $y\ne \lambda$, $|xy|\le \mpl(L)$, and~$xy^iz\in L$ for each $i \in \N$.
	In particular, we have~$xy^0z\in L$. On the other hand, we see that~$|xy^0z|=|xz|<|w|$. Since~$w'$ is the longest scattered factor of~$w$ in~$L$ we derive that~$|xz|\le |w'|$. 
	Therefore, we obtain that~$|y|=|w|-|xz|\ge |w|-|w'|>\mpl(L)$ which contradicts the assumption~$|y|\le \mpl(L)$ and therefore proves the stated claim.
	\iflipics\else\qed\fi
\end{proof}
\fi

Since~$\mpl(L)<\infty$ for all regular languages~$L$, we derive that~$\sd(L)<\infty$ for all regular languages~$L$, too.
Although we have~$\mpc(L)\le \mpl(L)$ and~$\sd(L)\le \mpl(L)$ for all languages~$L$ the measures~$\mpc$ and~$\sd$ are incomparable (even for regular languages).

\ifpaper
\begin{theorem}{lem:sd-mpc-incomparable}
	Let $p_1, p_2 \in \N$. There is a regular language~$L$ such that~$\sd(L)=p_1$ and~$\mpc(L)=p_2$.
\end{theorem}
\begin{proof}
	We distinguish the cases~$p_1\ge p_2$ and~$p_2> p_1$.
	\begin{itemize}
		\item If $p_1\ge p_2$ we define~$L=(\ta^{p_1})^* \ta^{p_2-1}$.
		Observe that each word in~$L$ is $\mpc$-pumpable except the word~$\ta^{p_2-1}$ and thus~$\mpc(L)=|\ta^{p_2-1}|+1=p_2-1+1=p_2$. 
		
		On the other hand, we observe that each word~$w$ in~$L$ is equal to~$\ta^{p_2-1+p_1\cdot j}$ for a natural number~$j$.
		Hence, the set of the scattered factors of~$w$ which are in~$L$ is
		$$\ScatFact(w)\cap L=\{\ta^{p_2-1+p_1\cdot k}\mid 0\le k\le j\}$$
		and thus~$\sd(L)=p_1$.
		
		\item In the case $p_2> p_1$ set~$L=(\ta^{p_1})^*+\tb^{\le p_2-1}$.
		The reader may confirm that the word~$\tb^{p_2-1}$ is not $\mpc$-pumpable while all words of length at least~$p_2$ are in~$(\ta^{p_1})^+\subset L$ and are pumpable by their prefix~$\ta^{p_1}$.
		Therefore, we have~$\mpc(L)=p_2$.
		
		In addition, for each word~$w=(\ta^{p_1})^j \in (\ta^{p_1})^+\subset L$ we have
		$$\ScatFact(w)\cap L=\{\ta^{k\cdot p_1}\mid 0\le k\le j\}$$
		which implies~$\sd(\ScatFact(w)\cap L)=p_1$. 
		All non-empty words in~$\tb^{\le p_2-1}$ contain a scattered factor in~$L$ that contains exactly one letter less.
		Thus, we obtain that~$\sd(L)=p_1$.\iflipics\else\qed\fi
	\end{itemize}	
\end{proof}
\else
\begin{theorem}{lem:sd-mpc-incomparable}
	Let $p_1, p_2 \in \N$. There is a regular language~$L$ such that~$\sd(L)=p_1$ and~$\mpc(L)=p_2$.
\end{theorem}
\begin{proof}
	We distinguish the cases~$p_1\ge p_2$ and~$p_2> p_1$.
	\begin{itemize}
		\item If $p_1\ge p_2$ we define~$L=(\ta^{p_1})^* \ta^{p_2-1}$.
		Observe that each word in~$L$ is $\mpc$-pumpable except the word~$\ta^{p_2-1}$ and thus~$\mpc(L)=|\ta^{p_2-1}|+1=p_2-1+1=p_2$. 
		
		On the other hand, we observe that each word~$w$ in~$L$ is equal to~$\ta^{p_2-1+p_1\cdot j}$ for a natural number~$j$.
		Hence, the set of the scattered factors of~$w$ which are in~$L$ is
		$$\ScatFact(w)\cap L=\{\ta^{p_2-1+p_1\cdot k}\mid 0\le k\le j\}$$
		and thus~$\sd(L)=p_1$.
		
		\item In the case $p_2> p_1$ set~$L=(\ta^{p_1})^*+\tb^{\le p_2-1}$.
		The reader may confirm that the word~$\tb^{p_2-1}$ is not $\mpc$-pumpable while all words of length at least~$p_2$ are in~$(\ta^{p_1})^+\subset L$ and are pumpable by their prefix~$\ta^{p_1}$.
		Therefore, we have~$\mpc(L)=p_2$.
		
		In addition, for each word~$w=(\ta^{p_1})^j \in (\ta^{p_1})^+\subset L$ we have
		$$\ScatFact(w)\cap L=\{\ta^{k\cdot p_1}\mid 0\le k\le j\}$$
		which implies~$\sd(\ScatFact(w)\cap L)=p_1$. 
		All non-empty words in~$\tb^{\le p_2-1}$ contain a scattered factor in~$L$ that contains exactly one letter less.
		Thus, we obtain that~$\sd(L)=p_1$.\iflipics\else\qed\fi
	\end{itemize}	
\end{proof}
\fi

Besides the aforementioned structural properties derived for non-regular languages which do \emph{not} fulfill Lemma~\ref{lem:mpc} or \ref{lem:mpl}, we obtain properties of their witness languages.

\begin{lemma}
	Let~$L$ be a language and $K \in \{\mpl,\mpc\}$. If~$W^{K}(L)=\{\lambda\}$, then either
	$L\in\mathcal{L}_{K}$ or~$K(L)=1$. Additionally, $W^{K}(L)=\emptyset$
	if and only if~$L=\emptyset$.
\end{lemma}

\begin{proof}
	We observe that the shortest word in~$L$ cannot be pumped w.r.t. any constant which directly implies the second statement.
	
	The property $W^{K}(L)=\{\lambda\}$ implies that each non-empty word~$w\in L$ can be pumped w.r.t. a constant~$p_w$ with~$p_w<K(L)$.
	If there are no non-empty words in~$L$ we have that~$L=\{\lambda\}=W^{K}(L)$ and thus we assume that~$L$ contains at least one non-empty word.
	Since~$K(L)=\infty$ implies~$L\in\mathcal{L}_{K}$, we assume that~$K(L)=p<\infty$. 
	For the sake of contradiction, assume that~$p>1$.
	Since~$p=\sup_{w\in L} \{p_w\}$ we derive that there must exist~$w\in L$ with $p_w>1$ implying that~$w$ is not pumpable w.r.t. $K(L)$. Therefore, we have $w\in W^{K}(L)$ which contradicts~$W^{K}(L)=\{\lambda\}$.
\end{proof}

Each word~$w=\ta^n \tb^n$ in the language~$L=\bigcup_{n\ge 0}\{\ta^n \tb^n\}^*$ can be pumped w.r.t. $2n$ by pumping its infix~$w$ but it cannot be pumped w.r.t. any smaller value. Hence, $\mpl(L)=\infty$ and thus~$L\in\mathcal{L}_{\mpl}$. Since for each~$n\ge 1$ we have that~$2n<\mpl(L)$ we obtain that each non-empty word is not a witness in $L$ and thus~$W^{\mpl}(L)=\{\lambda\}$.

	\section{Non-unary Witness Languages}\label{sec:non-unary-witness}
	Within this section, we start the study of witness languages, i.e., the words that are not $K$-pumpable w.r.t. $K \in \{\mpc, \mpl\}$ for some given non-unary regular language $L$. 
As an example, consider $L = \{w \in \{\ta,\tb\}^* \mid |w|_{\ta}\mbox{ even}\}$. 
Since we have the minimal pumping constant $\mpc(L) = 1$ we obtain $W(L) = \{\lambda\}$.

First, we are going to remark some facts on witness languages w.r.t. \Cref{lem:mpc} and $\mpc$. 

\ifpaper
\begin{lemma}\label{lem:witness-language-finite}
	Every witness language $W^{\mpc}(L)$ for a regular language $L$ is regular, actually finite, with no word being longer than $\mpc(L)-1$.
\end{lemma}
\begin{proof}
	We are going to show the claim by proving that every word within $W^{\mpc}(L)$ is of length strictly smaller than $\mpc$.
	Suppose that there exists an element that is larger than $\mpc$. Since $L$ is regular this implies that every word $w \in L$ of length $|w| \geq \mpc(L)$ is $\mpc$-pumpable meaning that $w$ cannot be a witness of $L$.
	Thus, $W^{\mpc}(L) \subseteq \Sigma^{\mpc(L)}$, i.e., $W^{\mpc}(L)$ is a finite and thus a regular language.\iflipics\else\qed\fi
\end{proof}
\else
\begin{lemma}\label{lem:witness-language-finite}
	Every witness language $W^{\mpc}(L)$ for a regular language $L$ is regular, actually finite, with no word being longer than $\mpc(L)-1$.
\end{lemma}
\begin{proof}
	We are going to show the claim by proving that every word within $W^{\mpc}(L)$ is of length strictly smaller than $\mpc$.
	Suppose that there exists an element that is larger than $\mpc$. Since $L$ is regular this implies that every word $w \in L$ of length $|w| \geq \mpc(L)$ is $\mpc$-pumpable meaning that $w$ cannot be a witness of $L$.
	Thus, $W^{\mpc}(L) \subseteq \Sigma^{\mpc(L)}$, i.e., $W^{\mpc}(L)$ is a finite and thus a regular language.\iflipics\else\qed\fi
\end{proof}
\fi

Next, we show that every finite language is a witness language w.r.t.~$\mpc$ and~$\mpl$.

\ifpaper
\begin{replemma}{lem:finite-language-its-own-witness}
	Let $K \in \{\mpl, \mpc\}$. For finite languages $L$ we have $W^K(L) = L$.
\end{replemma}
\else
\begin{replemma}{lem:finite-language-its-own-witness}
	Let $K \in \{\mpl, \mpc\}$. For finite languages $L$ we have $W^K(L) = L$.
\end{replemma}
\begin{proof}
	Since we have $w \in L$ for all $w \in W^K(L)$ by definition, the direction $W^K(L) \subseteq L$ is immediate.
	
	Further, let $w \in L$. Since $L$ is finite and $\mpc(L) = \mpl(L) = \max\{|w| \mid w \in L\} +1$ we have that no $w \in L$ is $K$-pumpable. Thus, $w \in W^K(L)$.\iflipics\else\qed\fi
\end{proof}
\fi

The definition of~$\mpc$ implies that each regular language~$L$ contains a word of length $\mpc(L)-1$ that cannot be pumped as well as that all words in $L$ of length at least $\mpc(L)$ are pumpable (which is not necessarily the case for $\mpl$). As a direct consequence, we obtain the following.

\begin{corollary}\label{cor:mpc-WL-leq-mpc-L}
	For each regular language~$L$ we have~$\mpc(W^{\mpc}(L))= \mpc(L)$.\qed
\end{corollary}

The following lemma shows that it is not possible to prove a similar or at least weaker variant of the previous corollary for~$\mpl$.

\ifpaper
\begin{replemma}{lem:mpl-distance-arbitrarily-large}
	Let $k,p \in \N$ with $k\ge p\ge 2$. There is a unary regular language~$L$ such that~$\mpl(L)=p$ and~$\mpl(W^{\mpl}(L))=k$.
\end{replemma}
\else
\begin{replemma}{lem:mpl-distance-arbitrarily-large}
	Let $k,p \in \N$ with $k\ge p\ge 2$. There is a unary regular language~$L$ such that~$\mpl(L)=p$ and~$\mpl(W^{\mpl}(L))=k$.
\end{replemma}
\begin{proof}
%
	
	Suppose first that~$k=p$, and let~$L=\{\ta^{p-1}\}$.
	Since~$L$ is finite and its longest word has length~$p-1$, we have~$\mpl(L)=p$.
	Moreover, Lemma~\ref{lem:finite-language-its-own-witness} yields
	$W^{\mpl}(L)=L$, and hence~$\mpl\bigl(W^{\mpl}(L)\bigr)=p=k$.
	
	For the remainder of the proof, assume that~$k>p$. Set
	$
	r=(k-1)\bmod p$ and
	$m=\left\lfloor\frac{k-1}{p}\right\rfloor.
	$
	Then~$
	k-1=r+mp$, $0\leq r<p$, and
	$m\geq 1.$
	Consider the unary regular language
	\[
	L
	=
	\ta^r(\ta^p)^*
	\cup
	\{\ta^n\mid n\geq k\}.
	\]
	Thus, the lengths of the words in~$L$ that are smaller than~$k$
	are precisely
	\[
	r,r+p,\ldots,r+mp=k-1.
	\]
	We observe that each word shorter than~$k$ and longer than~$r<p$ can be pumped by~$\ta^p$ while all words in~$\ta^{\ge k}$ can be pumped by~$\ta$. Therefore, $\mpl(L)\le p$.
	On the other side, we see that all words shorter than~$k$ cannot be pumped by a factor shorter than~$\ta^p$ and thus~$\mpl(L)=p$.
	Moreover, we obtain that
	\(
	W^{\mpl}(L)
	=
	\{\ta^{r+jp}\mid 0\leq j\leq m\}
	=
	\ta^r(\ta^p)^{\leq m}.
	\)
	Finally, $W^{\mpl}(L)$ is finite, and its longest word has length~$r+mp=k-1$ and thus
	$$\mpl(W(L))=1+\max\{\,|w|\mid w\in W(L)\,\}=1+(k-1)=k,$$
	which proves the stated claim.
	\iflipics\else\qed\fi
\end{proof}
\fi

It is not possible to omit the condition~$k\ge p$ in the previous lemma because for each language~$L$ and each word~$w \in W^{\mpl}(L)$ we know that~$w$ is not $\mpl$-pumpable in~$L$ w.r.t. each value that is strictly less than~$p$. Since~$W^{\mpl}(L)\subseteq L$, the same property must also apply to~$W^{\mpl}(L)$.

From Lemmata~\ref{lem:witness-language-finite} and~\ref{lem:finite-language-its-own-witness} we obtain that~$W^{\mpc}(L)= W^{\mpc}(W^{\mpc}(L))$ for all regular languages~$L$.
This is not true if we investigate witness languages w.r.t.~$\mpl$.

\ifpaper
\begin{replemma}{lem:mpl-W-L-interval-2mpl-L}
	There is a regular language~$L$ such that~$W^{\mpl}(L)\ne W^{\mpl}(W^{\mpl}(L))$.
\end{replemma}
\else
\begin{replemma}{lem:mpl-W-L-interval-2mpl-L}
	There is a regular language~$L$ such that~$W^{\mpl}(L)\ne W^{\mpl}(W^{\mpl}(L))$.
\end{replemma}
\begin{proof}
	We set~$L=\mathtt{a}^*\cup (\mathtt{b}^p)^+\mathtt{a}^{k-p-1}\cup(\mathtt{c}^{p})^*$ for some $k \geq p \geq 2$. Observe that each word in~$\mathtt{a}^+$ can be pumped by
	its first letter
	while 
	the words in the language $(\tb^p)^+\mathtt{a}^{k-p-1}\cup(\mathtt{c}^{p})^*$ are pumpable by their prefix~$\mathtt{b}^p$ and~$\mathtt{c}^p$, \resp, because~$(\mathtt{b}^p)^0\mathtt{a}^{k-p-1}=\mathtt{a}^{k-p-1}$ is in~$L$ as well.
	Additionally, none of these words can be pumped by shorter factors.
	Since~$L=\mathtt{a}^+\cup (\mathtt{b}^p)^+\mathtt{a}^{k-p-1}\cup(\mathtt{c}^{p})^*$ the value~$\mpl(L)$ is equal to the length of the longest of the aforementioned factors.
	Thus, we have~$\mpl(L)=p$.
	Since~$p\ge 2$ none of the words in the language~$\mathtt{a}^+$ is in~$W^{\mpl}(L)$ and thus~$W^{\mpl}(L)=(\mathtt{b}^p)^+\mathtt{a}^{k-p-1}\cup(\mathtt{c}^{p})^*$.
	Again all words in~$(\mathtt{b}^p)^{\ge 2}\mathtt{a}^{k-p-1}\cup(\mathtt{c}^{p})^*\subset W^{\mpl}(L)$ are pumpable by their prefixes~$\mathtt{b}^p$ and~$\mathtt{c}^p$, \resp. In turn, this does not hold for~$\mathtt{b}^p\mathtt{a}^{k-p-1}\in W^{\mpl}(L)$ because none of its proper factors is in~$W^{\mpl}(L)$ except~$\lambda$ and we have~$(\mathtt{b}^p\mathtt{a}^{k-p-1})^2\notin W^{\mpl}(L)$ which implies that the word~$\mathtt{b}^p\mathtt{a}^{k-p-1}$ cannot be pumped by any of its factors.
	Hence, we have~$\mpl(W^{\mpl}(L))=k+p-p-1+1=k$ and thus~$W^{\mpl}(W^{\mpl}(L))=\mathtt{b}^p\mathtt{a}^{k-p-1}\cup \{\lambda\}$.
	\iflipics\else\qed\fi
\end{proof}
\fi

\ifpaper
The language used to prove the previous lemma is~$L=\ta^*\cup (\tb^{p})^+\ta^{k-p-1}\cup(\tc^{p})^*$ which fulfils~$W^{\mpl}(L)=(\tb^p)^+\mathtt{a}^{k-p-1}\cup(\tc^{p})^*$, $\mpl(L)=p$ and $W^{\mpl}(W^{\mpl}(L))=\mathtt{b}^p\mathtt{a}^{k-p-1}\cup \{\lambda\}$.
We observe that
$$\ScatFact(\tb^p\ta^{k-p-1})\cap W^{\mpl}(L)=\{\tb^p\ta^{k-p-1},\lambda\}$$ while each other word~$w \in W^{\mpl}(L)$ contains a factor which is~$p$ letters shorter than~$w$. If $k \ge p$, we have~$\sd(W^{\mpl}(L))=k-1$ and thus by changing~$k$ accordingly we obtain the following corollary.

\begin{repcorollary}{cor:mpl-sd-infty}
	Let $k,p \in \N$ with $k\ge p\ge 2$, there is a regular language~$L$ such that~$\mpl(L)=p$ and~$\sd(W^{\mpl}(L))=k$.\qed
\end{repcorollary}
\else
For the language~$L$ used in the previous proof, we observe that~$\ScatFact(\tb^{p-1}\ta^{k-p-1})\cap W^{\mpl}(L)=\{\lambda\}$ while each other word~$w$ in~$W^{\mpl}(L)$ contains a factor that is either~$p-1$ or~$p$ letters shorter than~$w$. Hence, we have~$\sd(W^{\mpl}(L))=k$ and we obtain the following corollary.

\begin{repcorollary}{cor:mpl-sd-infty}
	Let $k,p \in \N$ with $k\ge p\ge 2$, there is a regular language~$L$ such that~$\mpl(L)=p$ and~$\sd(W^{\mpl}(L))=k$.\qed
\end{repcorollary}

\fi

The following proposition classifies the family of witness languages w.r.t. $\mpc$ completely.

\ifpaper
\begin{proposition}\label{prop:every-set-is-witnesslanguage}
	For every $p \in \N_0$ we have $2^{\Sigma^{\le p-1}}\setminus\emptyset = \mathcal{W}_p^{\mpc,u}$ as well as $2^{\Sigma^{\le p-1}}\setminus\emptyset = \mathcal{W}_p^{\mpc}$.
\end{proposition}

\begin{proof}
	For the first direction let $L \in \mathcal{W}^{\mpc}_p$. Thus, $L$ is a witness language for $p$ w.r.t. $\mpc$. By \Cref{lem:witness-language-finite} we obtain that $L \in 2^{\Sigma^{\leq p-1}}$. Further, we know that $L \neq \emptyset$ since there does not exist any factorization for $\lambda$, i.e., $\lambda \in L$. 
	
	For the other direction, let $L \in 2^{\Sigma^{\leq p-1}} \setminus \emptyset$. To show that $L$ is a witness language we show that there exists a regular language $L'$ such that $W^{\mpc}(L') = L$. By \Cref{lem:finite-language-its-own-witness} we know that $W^{\mpc}(L') = L' = L$, i.e., $L \in \mathcal{W}^{\mpc}_p$. The claim follows.\iflipics\else\qed\fi
\end{proof}
\else
\begin{proposition}\label{prop:every-set-is-witnesslanguage}
	For every $p \in \N_0$ we have $2^{\Sigma^{\le p-1}}\setminus\emptyset = \mathcal{W}_p^{\mpc,u}$ as well as $2^{\Sigma^{\le p-1}}\setminus\emptyset = \mathcal{W}_p^{\mpc}$.
\end{proposition}

\begin{proof}
	For the first direction let $L \in \mathcal{W}^{\mpc}_p$. Thus, $L$ is a witness language for $p$ w.r.t. $\mpc$. By \Cref{lem:witness-language-finite} we obtain that $L \in 2^{\Sigma^{\leq p-1}}$. Further, we know that $L \neq \emptyset$ since there does not exist any factorization for $\lambda$, i.e., $\lambda \in L$. 
	
	For the other direction, let $L \in 2^{\Sigma^{\leq p-1}} \setminus \emptyset$. To show that $L$ is a witness language we show that there exists a regular language $L'$ such that $W^{\mpc}(L') = L$. By \Cref{lem:finite-language-its-own-witness} we know that $W^{\mpc}(L') = L' = L$, i.e., $L \in \mathcal{W}^{\mpc}_p$. The claim follows.\iflipics\else\qed\fi
\end{proof}
\fi

Although the situation for $\mpc$ was easy, we cannot obtain the same result for $\mpl$ (\Cref{lem:mpl}). For $\mpl$ the witness language does not need to be finite, e.g., the language~$L=(\ta^n)^*$ is equal to its witness language because none of the words in~$L$ can be pumped w.r.t. Lemma~\ref{lem:mpl} for any value less than~$\mpl(L)=n$.
However, the fact that all finite languages~$L$ fulfil~$\mpc(L)=\mpl(L)$ implies the following result.

\begin{lemma}\label{lem:mpl-finite-languages-bounded}
	For every $p\in\N$ we have~$2^{\Sigma^{\le p-1}}\setminus\emptyset \subset \mathcal{W}_p^{\mpl}.$\qed
\end{lemma}

Nevertheless, we will prove that for each regular language its witness language w.r.t. Lemma~\ref{lem:mpl} is regular as well. Therefore, we use the following result, which is a direct consequence of the proof of \cite[Theorem 8]{GrHoRa24b}.

\begin{corollary}\label{cor:small-words-are-sufficient-for-mpl}
	Let $L$ be a regular language over~$\Sigma$ accepted by an~$n$-state finite automaton and let~$p\in\N$. For each word~$w\in L$ with~$|w|>n+|M(L)| \geq p$ which can be pumped w.r.t. Lemma~\ref{lem:mpl} and the value~$p$ there is a word~$w'\in L$ such that\\
	- $p\le |w'|\le n+|M(L)|$,\\
	- the prefixes of length~$p$ are equal for~$w$ and~$w'$,\\
	- $\varphi_L(w)=\varphi_L(w')$.\qed
\end{corollary}

\ifpaper
\begin{reptheorem}{thm:Reg-L-implies-wit-mpl-L_is-regular}
	For any regular language~$L$, $W^{\mpl}(L)$ is regular.
\end{reptheorem}

\begin{proof}
	In the cases~$L=\emptyset$ and $\mpl(L)=1$ we obtain directly~$W^{\mpl}(L)=L$. 
	Therefore, assume that~$L$ is accepted by an~$n$-state \dfa\ with~$n>2$.
	Then $1\le\mpl(L)<n$.
	For each word $w\in\Sigma^*$, write $w=uv$ with
	$|u|=\min\{n,|w|\}$, and define its type by
	\[
	\tau(w)=(u,\varphi_L(v)).
	\]
	As the monoid~$M(L)$ is finite and~$|u|\le n$, there are only finitely many types.
	
	We show that membership in $W^{\mpl}(L)$ depends only on this
	type. Suppose that $\tau(w_1)=\tau(w_2)$.
	If $|w_1|<n$, then $w_1=w_2$, and the assertion is immediate.
	Otherwise, write
	\[
	w_1=uv_1,\qquad w_2=uv_2,\qquad |u|=n,
	\qquad \varphi_L(v_1)=\varphi_L(v_2).
	\]
	In particular, $w_1\in L$ if and only if $w_2\in L$.
	By using Corollary~\ref{cor:small-words-are-sufficient-for-mpl} for~$p=\mpl(L)-1$
	we obtain that for every decomposition $u=xyz$ with $y\neq\lambda$
	and $|xy|\leq \mpl(L)-1$, we have
	\[
	xy^tzv_1\in L
	\quad\Longleftrightarrow\quad
	xy^tzv_2\in L
	\qquad\text{for every }t\geq0.
	\]
	Since~$|xy|\le\mpl(L)-1\le|xyz|=|u|$, it follows that
	\[
	w_1\in W^{\mpl}(L)
	\quad\Longleftrightarrow\quad
	w_2\in W^{\mpl}(L).
	\]
	
	Finally, the type after reading an additional letter depends
	only on the current type. More precisely, with $1_{M(L)}$ denoting
	the identity of~$M(L)$, define
	\[
	(u,m)\cdot \ta=
	\begin{cases}
		(u\ta,1_{M(L)}), & |u|<n,\\
		(u,m\varphi_L(\ta)), & |u|=n.
	\end{cases}
	\]
	These transitions, together with the initial state
	$(\lambda,1_{M(L)})$, define a DFA on the finite set of types.
	In this regard, its set of accepting states consists of all types
	that belong to~$W^{\mpl}(L)$.
	The resulting DFA accepts exactly $W^{\mpl}(L)$, which is
	therefore regular.
	\iflipics\else\qed\fi
\end{proof}

\else
\begin{reptheorem}{thm:Reg-L-implies-wit-mpl-L_is-regular}
	For any regular language~$L$, $W^{\mpl}(L)$ is regular.
\end{reptheorem}

\fi

Note that although the previous proof is constructive, the construction is very likely not feasible in sub-exponential time.

\ifpaper
\begin{reptheorem}{thm:exponential-size-growing-alphabet}
	Let~$K\in\{\mpc,\mpl\}$. For each~$n\in \N_{\ge 3}$ there is a \dfa~$A_n$ such that~$\dsc(L(A))=n+1$ and~$\dsc(W^K(L(A)))\in \Omega(2^n)$.
\end{reptheorem}

The previous result can be shown by using an~$n$-state DFA over an input alphabet of size~$n^2-n+1$ with~$n-1$ accepting states and one non-accepting sink state. Each of the input letters maps exactly one of the accepting states onto another accepting state while all other states are mapped onto the sink state.
By encoding the alphabet into a binary alphabet \textit{via} a construction introducing $\Theta(n^3)$ states, we obtain the following stronger result.

\else
\begin{reptheorem}{thm:exponential-size-growing-alphabet}
	Let~$K\in\{\mpc,\mpl\}$. For each~$n\in \N_{\ge 3}$ there is a \dfa~$A_n$ such that~$\dsc(L(A))=n+1$ and~$\dsc(W^K(L(A)))\in \Omega(2^n)$.
\end{reptheorem}
\begin{proof}
	
	Let~$A_n=\langle \{q_0,q_1,\ldots, q_{n-1},q_s\},\Sigma_n,\cdot, q_0, \{q_0,q_1,\ldots, q_{n-1}\}\rangle$ where~$\Sigma_n=\{\ta_0,\mathtt{a}_1,\mathtt{a}_2,\ldots,$ $\mathtt{a}_{n^2-n}\}$ and we set 
	\begin{align*}
		q_i\cdot \mathtt{a}_{(j-1)^2-j+2i+1}&=q_j, &\mbox{for~$0\le i<j\le n-1$,}\\
		q_j\cdot \mathtt{a}_{(j-1)^2-j+2i+2}&=q_i, &\mbox{for~$0\le i<j\le n-1$,}
	\end{align*}
	where all not-explicitly defined transitions are transitions to the sink state~$q_s$.
	The automaton $A_n$ is depicted in Figure~\ref{fig:exponential-sc}.
	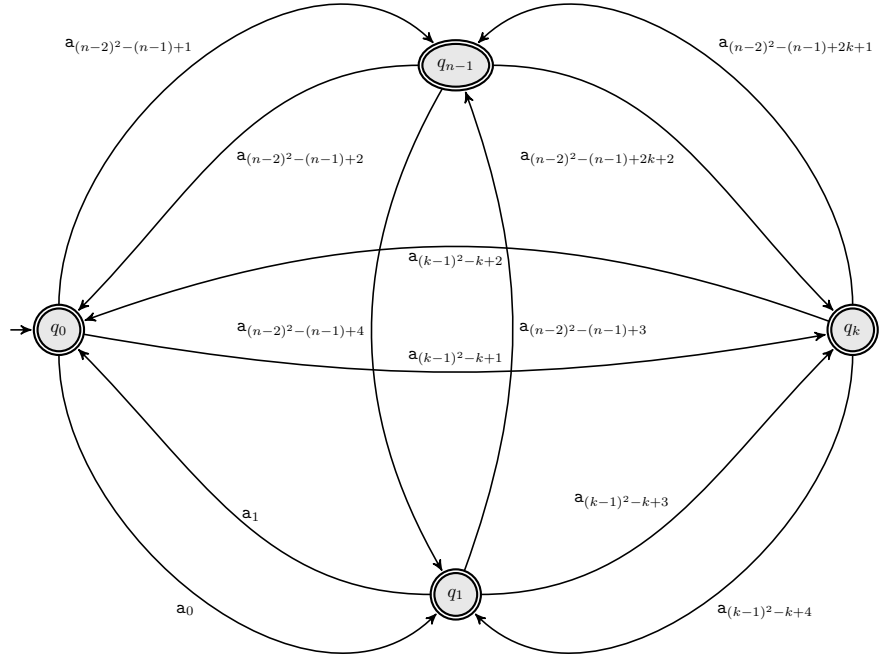
\begin{figure}[h!]
		\begin{center}
						\begin{tikzpicture}[->,>=stealth',shorten >=1pt,auto,node distance=2.9cm,
	semithick,scale=0.7, every node/.style={scale=0.7},initial text=]
	\tikzstyle{every state}=[fill={rgb:black,1;white,10},ellipse]

	\node[state,initial,accepting] (q0) at (180:7.5cm) [thick] {$q_0$};	
	\node[state,accepting] (q1) at (270:5cm) [thick] {$q_1$};
	\node[state,accepting] (q2) at (0:7.5cm) [thick] {$q_k$};	
	\node[state,accepting] (qn1) at (90:5cm) [thick] {$q_{n-1}$};	



	\path 

	(q0) edge[in=225,out=270] node[below left]{$\ta_0$} (q1)
	(q1) edge[in=315,out=180] node[above, right]{$\ta_1$} (q0)

	(q2) edge[in=315,out=270] node[below right]{$\ta_{(k-1)^2-k+4}$} (q1)
	(q1) edge[in=225,out=0] node[above left]{$\ta_{(k-1)^2-k+3}$} (q2)

	(qn1) edge[in=45,out=180] node[below right]{$\ta_{(n-2)^2-(n-1)+2}$} (q0)
	(q0) edge[in=135,out=90] node[above left]{$\ta_{(n-2)^2-(n-1)+1}$} (qn1)

	(qn1) edge[in=135,out=0] node[below left]{$\ta_{(n-2)^2-(n-1)+2k+2}$} (q2)
	(q2) edge[in=45,out=90] node[above right]{$\ta_{(n-2)^2-(n-1)+2k+1}$} (qn1)
	
	(q2) edge[bend right=20] node[below]{$\ta_{(k-1)^2-k+2}$} (q0)
	(q0) edge[bend right=10] node[above]{$\ta_{(k-1)^2-k+1}$} (q2)

	(qn1) edge[bend right] node[left]{$\ta_{(n-2)^2-(n-1)+4}$} (q1)
	(q1) edge[bend right=20] node[right]{$\ta_{(n-2)^2-(n-1)+3}$} (qn1)

	;
\end{tikzpicture}
		\end{center}
		\vspace{-2cm}
		\caption{The minimal \dfa~$A_n$ with~$n+1$ states for which we have~$\dsc(W^K(L(A_n)))\in \Omega(2^n)$, where 
			the non-accepting sink state~$q_s$ and all transitions to it are not displayed.
		}\vspace{-0.1cm}
		\label{fig:exponential-sc}
	\end{figure}
	%
	The reader may verify that for each letter~$\ta\in\Sigma_n$ there is exactly one state which is not mapped onto the non-accepting state by~$\mathtt{a}$.
	Hence, we directly obtain that the automaton~$A_n$ is minimal.
	
	
	\begin{mialc}
		The language~$W^K(L(A_n))$ consists of all words which induce a simple path\footnote{A \emph{path} on a graph is a sequence of nodes where two consecutive nodes are connected in the graph. A path is called \emph{simple} if the path does not contain a node more than once.} on~$A_n$ not containing~$q_s$.
	\end{mialc}
	
	\begin{proof}
		Given a word~$w\in L$ which induces a simple path~$p$ on~$A_n$ for which we assume that~$w\notin W^K(L)$, i.e., $w$ is pumpable.
		Any decomposition~$w=xyz$ with~$|y|> 0$ gives rise to three states
		$$q_x= q_0\cdot x,\quad q_{xy}=q_x\cdot y,\quad\mbox{and}\quad q_{xyz}=q_{xy}\cdot z,$$ 
		where all states are pairwise different, since~$p$ is a simple path.
		Let~$\mathtt{a}$ be the first letter of~$z$ then for all states~$q$ of~$A_n$ unequal to~$q_x$ we have~$q\cdot \mathtt{a}=q_s$.
		If~$w$ is pumpable by~$y$ we obtain that~$q_{x}=q_{xy^i}$ for each $i \in \N$ which in turn implies that~$y$ induces a self-loop on~$q_x$.
		Indeed, this contradicts the assumption that~$p$ is a simple path and thus~$w$ cannot be pumped w.r.t. to any decomposition~$w=xyz$ with~$|y|>0$.
		
		On the other hand, each word~$w\in L$ which does not induce a simple path on~$L$ can be pumped by the factor of~$w$ which induces a cycle.\iflipics\else\qed\fi
	\end{proof}
	
	It remains to show that the minimal \dfa\ accepting the witness language of~$L(A_n)$ is exponentially large.
	
	\begin{mialc}\label{clm:ex-WL-general}
		The language~$\dsc(W^K(L(A_n)))\in\Omega(2^n)$.
	\end{mialc}
	
	\begin{proof}
		Since the stated claim is trivial for~$n=1,2,3$ we assume without loss of generality that~$n\ge 4$.
		
		For proving the stated claim we show~$\dsc(W^K(L(A_{n+1})))\ge 2\cdot \dsc(W^K(L(A_{n})))$.
		Observe that~$W^K(L(A_{n}))\subseteq W^K(L(A_{n+1}))$ as the automaton~$A_{n+1}$ can be obtain by adding the state~$q_{n+1}$ and the letters
		$$\mathtt{a}_{(n-1)^2-n+1},\mathtt{a}_{(n-1)^2-n+2},\ldots, \mathtt{a}_{(n-1)^2-n+2(n-1)+2}$$
		to~$A_n$.
		Therefore, let~$B_n$ be the minimal \dfa\ accepting the language~$W^K(L(A_n))$.
		
		First we observe that the words of length~$n-1$ leave the automaton~$B_n$ in the same accepting state which is mapped by each letter onto the non-accepting sink state.
		Hence, let~$w,w'\in W^K(L(A_{n+1}))$ be words of length~$k\le n-2$ such that~$w$ contains none of the letters in~$\Sigma_{n+1}\setminus \Sigma_n$ while~$w'$ contains such a letter.
		There is a letter~$\mathtt{a}\in\Sigma_{n+1}\setminus \Sigma_n$ such that~$w\mathtt{a}$ maps~$q_0$ onto~$q_n$. Then, $w'\mathtt{a}$ either maps~$q_0$ onto~$q_n$ and thus induces a non-simple path or~$w'\mathtt{a}$ maps~$q_0$ onto~$q_s$.
		In both cases~$w'\mathtt{a}$ is not in~$W^K(L(A_{n+1}))$.
		Therefore, the states~$q_w$ and~$q_{w'}$ are mapped by~$\mathtt{a}$ onto the states~$q_{w\mathtt{a}}$ and~$q_{w'\mathtt{a}}$, \resp, where the first one is accepting while the latter one is non-accepting.
		Thus, $q_w$ and~$q_{w'}$ cannot be equivalent.
		Since for each word~$w\in W^K(L(A_n))$ of length at most~$n-2$ such a letter~$\mathtt{a}\in\Sigma_{n+1}\setminus\Sigma_n$ exists, we obtain that~$B_{n+1}$ must contain at least twice as many states as~$B_n$.
		\iflipics\else\qed\fi
	\end{proof}
	As a consequence, of the previous claims and the minimality of~$A_n$ the statement of the theorem follows.\iflipics\else\qed\fi
\end{proof}

A schematic drawing of the minimal \dfa~$A_n$ with~$n+1$ states for which we have~$\dsc(W^K(L(A_n)))\in \Omega(2^n)$ can be found in \Cref{fig:exponential-sc}.
We will enhance the previous result by encoding the alphabet of quadratic size into a binary alphabet.
\fi

\ifpaper
\begin{reptheorem}{thm:exponential-size-binary}
	Let~$K\in\{\mpc,\mpl\}$. There exists a $n \in \N$ such that there is a binary language~$L$ with~$\dsc(L)=n$ and~$\dsc(W^K(L))\in\Omega(2^n)$.
\end{reptheorem}
\else
\begin{reptheorem}{thm:exponential-size-binary}
	Let~$K\in\{\mpc,\mpl\}$. There exists a $n \in \N$ such that there is a binary language~$L$ with~$\dsc(L)=n$ and~$\dsc(W^K(L))\in\Omega(2^n)$.
\end{reptheorem}
\begin{proof}
	Since the construction is straight-forward we will sketch the construction and prove the non-trivial properties instead.
	From the automaton~$A_n$ as defined in the proof of Theorem~\ref{thm:exponential-size-growing-alphabet} we define~$A_n'$ by adding~$n$ states for each transition of~$A_n$, i.e., we the add~$(n^2-n-1)\cdot n$ states
	$$\{q_{i,j}\mid 0\le i\le n-1, 1\le j\le n^2-n-1\}$$
	to the state set of~$A_n$.
	We define the transition function of~$A_n'$ as follows, whereby we want the sequence~$\mathtt{0}^{r-1}\mathtt{1}\mathtt{0}^{n^2-n-1-(r-1)}$ to encode the letter~$\ta_r$ for $r \in [n^2-n]_0$.
	Assume that~$q_i\cdot_{A_n} \ta_r=q_j$ for~$0\le i,j\le n-1$ and $r \in [n^2-n]_0$. Then we set
	\begin{align*}
		q_i\cdot_{A_n'} \mathtt{0}&=q_{i,1},&\\
		q_i\cdot_{A_n'} \mathtt{1}&=q_s,&\\
		q_{i,k}\cdot_{A_n'} \mathtt{0}&=q_{i,k+1}, &\mbox{for~$1\le k\le n^2-n-2$ and~$k\ne i-1$,}\\
		q_{i,k}\cdot_{A_n'} \mathtt{0}&=q_s, &\mbox{for~$1\le k\le n^2-n-2$ and~$k\ne i-1$},\\
		q_{i,i-1}\cdot_{A_n'} \mathtt{1}&=q_{i,i},&\\
		q_{i,i-1}\cdot_{A_n'} \mathtt{0}&=q_s,&\\
		q_{i,n^2-n-1}\cdot_{A_n'} \mathtt{0}&=q_{i+1},&\\
		q_{i,n^2-n-1}\cdot_{A_n'} \mathtt{1}&=q_s.&
	\end{align*} 
	
	\begin{mialc}
		The automaton~$A_n'$ is minimal.
	\end{mialc}
	
	\begin{proof}
		Similar as for the automaton~$A_n$ we obtain that none of the states~$q_0,q_1,\ldots q_{n-1}, q_s$ is equivalent to another one of the aforementioned ones by exchanging each word~$w$ by its binary encoding of length~$|w|\cdot (n^2-n)$, i.e., each letter of~$w$ is exchanged by its binary encoding.
		Since there is a cycle over all states excluding~$q_s$ we directly obtain that states whose distance\footnote{Given a \dfa~$A$ and two states~$q,q'$ of~$A$. The \emph{distance} of~$q$ and~$q'$ is the length of the shortest word mapping~$q$ onto~$q'$. If this word does not exist we set the distance of~$q$ and~$q'$ to be~$\infty$.} is divisible by~$n^2-n$ cannot be equivalent.
		In addition, no state~$q_k$ can be equivalent to a state~$q_{i,k}$ for~$0\le i,j\le n-1$ and~$1\le k\le n^2-n-1$ because the first one is accepting while the second one is not.
		Therefore, it follows with the same argument as used before that states with an arbitrary distance cannot be equivalent and thus~$A_n'$ is minimal.\iflipics\else\qed\fi
	\end{proof}
	
	It remains to show that the binary encoding of the input alphabet does \emph{not} effect the exponential size of the automaton accepting the witness language.
	As~$A_n'$ is constructed from~$A_n$, we see that~$B_n'$ can be obtained from~$B_n$ in the same way.
	Hence, the inequivalence of~$2^n$ states of~$B_n'$ is shown be combining the proof of the previous claim and Claim~\ref{clm:ex-WL-general}.
	In summary the automaton~$A_n'$ has~$n'=(n^2-n-1)\cdot n+n+1=n^3-n^2-n+1$ while $B_n'$ has at least~$2^n\in\Omega(2^{n'})$ states.
	\iflipics\else\qed\fi
\end{proof}
\fi

	\section{The Witness Language Operator}\label{sec:witness-language-operator}
	Note that Section~\ref{sec:witness-language-operator} is an AI-assisted section. The ideas presented here are human-made but the proofs where assisted by ChatGPT 6 Astra.
We ensure that the usage of AI is restricted to Section~\ref{sec:witness-language-operator}.

By Lemma~\ref{lem:mpl-W-L-interval-2mpl-L}, the witness-language
operator is not idempotent: there exists a language~$L$ such that
\[
W^{\mpl}(L)\neq W^{\mpl}\bigl(W^{\mpl}(L)\bigr).
\]
This motivates the study of its iterated application. We show that
for every language the iteration reaches a fixed point by
stage~$\omega$ (which denotes the least finite ordinal), and that this bound is attained by the regular
language~$(\ta+\tb\tb)^*$. We then examine when the fixed point of a
regular initial language is itself regular.

Recall, that we use the convention
$\mpl(M)=\infty$ when no finite pumping constant exists. In this
case, the witness definition means that a word belongs to
$W^{\mpl}(M)$ precisely when it admits no pumping decomposition
in~$M$ with any finite prefix bound. This case will be relevant even
when all finite iterates of the initial language are regular.

For every language~$M$, we have
\begin{equation}
\label{eq:mpl-iteration-contracting}
W^{\mpl}(M)\subseteq M.
\end{equation}
Moreover, pumpability is monotone both in the prefix bound and in the
ambient language: a pumping decomposition in~$N\subseteq M$ is also
one in~$M$.

\begin{definition}[Pumping within a bounded prefix]
	\label{def:pump}
	For a language \(L\subseteq\Sigma^*\) and \(r\in\N_0\), put
	\[
	\Pump_r(L)=
	\left\{
	xyz\in L\ \middle|\
	\begin{array}{l}
		x,z\in\Sigma^*,\quad y\in\Sigma^+,\quad |xy|\leq r,\\
		xy^iz\in L\text{ for every }i\in\N_0
	\end{array}
	\right\}.
	\]
	Thus \(\Pump_0(L)=\emptyset\). Membership in \(\Pump_r(L)\)
	does not require the word to have length at least~\(r\).
\end{definition}

\begin{lemma}[Unpumpable words and shortest words]
\label{lem:mpl-iteration-unpumpable}
The following statements hold.
\begin{enumerate}
\item If no word of~$M$ admits a pumping decomposition in~$M$, then
      $W^{\mpl}(M)=M$. In particular, every finite language is a fixed point.
\item If $M\neq\emptyset$, every word of minimum length in~$M$
      belongs to~$W^{\mpl}(M)$.
\item If $\mpl(M)=\infty$, then $W^{\mpl}(M)$ is a fixed point.
\end{enumerate}
\end{lemma}

\begin{proof}
The first assertion follows from the definition. A finite language
admits no pumping decomposition, because pumping a non-empty factor
would produce words of unbounded length. A shortest word cannot be
pumped down: the exponent~$0$ would produce a shorter word of~$M$.
This proves the second assertion. Finally, when $\mpl(M)=\infty$,
the language $N=W^{\mpl}(M)$ consists of words unpumpable in~$M$.
These words remain unpumpable in~$N\subseteq M$, so the first assertion
applies to~$N$.
\end{proof}

\subsection{Transfinite iteration and stabilization}

\begin{definition}[Ordinals]
\label{def:mpl-ordinals}
An \emph{ordinal} is a transitive set well ordered by membership.
For ordinals~$\alpha,\beta$, we write $\alpha<\beta$ if
$\alpha\in\beta$. The successor of~$\alpha$ is
$\alpha+1=\alpha\cup\{\alpha\}$; a non-zero ordinal which is not a
successor is a \emph{limit ordinal}. The least infinite ordinal is
\[
\omega=\{0,1,2,\ldots\}.
\]
An ordinal is \emph{countable} if it admits an injection into~$\N_0$.
The least uncountable ordinal is denoted by~$\omega_1$; equivalently,
$\alpha$ is countable if and only if $\alpha<\omega_1$.
\end{definition}

\begin{remark}
\label{rem:mpl-countable-ordinals}
Countability does not imply $\alpha\leq\omega$. For example,
$\omega+1$, $\omega\cdot2$, and $\omega^2$ are countable and strictly
larger than~$\omega$. In particular,
\[
\omega+1=\omega\cup\{\omega\}=
\{0,1,2,\ldots,\omega\},
\]
whose greatest element~$\omega$ represents the first position after
all finite positions.
\end{remark}

\begin{definition}[Transfinite witness iteration]
\label{def:mpl-transfinite-witness-iteration}
For $L\subseteq\Sigma^*$, define
\[
L_0=L,\qquad L_{\alpha+1}=W^{\mpl}(L_\alpha),\qquad
L_\delta=\bigcap_{\beta<\delta}L_\beta
\]
for successor stages and non-zero limit ordinals~$\delta$, respectively.
The language
\[
L_\omega=\bigcap_{n\in\N_0}L_n
\]
is the \emph{$\omega$-limit} of the iteration.
\end{definition}

By \eqref{eq:mpl-iteration-contracting}, the iteration is descending:
$L_\beta\subseteq L_\alpha$ whenever $\alpha\leq\beta$.
Once two consecutive stages agree, all subsequent stages agree with
them, as follows by transfinite induction.

\begin{remark}
\label{rem:mpl-omega-limit-versus-fixed-point}
For a general contracting operator on languages, the intersection of
the finite iterates need not be a fixed point. In particular, one
cannot simply exchange the operator and a countable intersection.
For the witness operator, we will prove directly that $L_\omega$ is
always a fixed point; preservation of arbitrary countable
intersections is not needed.
\end{remark}

\begin{theorem}[Existence of an iterated witness fixed point]
\label{thm:mpl-witness-fixed-point}
Every language~$L\subseteq\Sigma^*$ has a least ordinal~$\alpha_L$
such that $L_{\alpha_L+1}=L_{\alpha_L}$, and $\alpha_L<\omega_1$.
We call~$\alpha_L$ the \emph{stabilization ordinal} and write
\[
\Fix_{\mpl}(L)=L_{\alpha_L}.
\]
Every shortest word of a non-empty~$L$ belongs to~$\Fix_{\mpl}(L)$.
In particular, this fixed point is non-empty whenever~$L$ is non-empty.
\end{theorem}

\begin{proof}
If $L_{\alpha+1}\subsetneq L_\alpha$ for every $\alpha<\omega_1$,
then the sets $L_\alpha\setminus L_{\alpha+1}$ would be non-empty
and pairwise disjoint. Fix an injection $e:\Sigma^*\to\N_0$.
Choosing the least $e$-value in each difference would give an
injection $\omega_1\to\N_0$, a contradiction. Thus, equality occurs
at a countable stage, and there is a least such stage. By the successor
rule, its language is a fixed point.

Let $u$ be a shortest word of~$L$. Transfinite induction shows that
$u\in L_\alpha$ at every stage: at a successor stage, it is still
a shortest word because $L_\alpha\subseteq L$, and survives by
Lemma~\ref{lem:mpl-iteration-unpumpable}; at a limit stage, it belongs
to the defining intersection.
\end{proof}

The countability argument uses only contraction. The stronger
bound~$\omega$ follows from the behaviour of the pumping lengths.

\begin{lemma}[Growth or stabilisation]
\label{lem:mpl-growth-or-stabilization}
Let $M\subseteq\Sigma^*$ satisfy $p=\mpl(M)<\infty$, and put
$N=W^{\mpl}(M)$ and $q=\mpl(N)\in\N_0\cup\{\infty\}$.
Then $q\geq p$. If $q=p$, then $W^{\mpl}(N)=N$.
\end{lemma}

\begin{proof}
For $p=0$, we have $M=N=\emptyset$. For $p=1$, the equality
$\operatorname{Pump}_0(M)=\emptyset$ gives $N=M$.
Assume $p\geq2$. Since $p-1$ is not a pumping length for~$M$,
there is a word
\[
w\in M\setminus\operatorname{Pump}_{p-1}(M),\qquad |w|\geq p-1.
\]
This word belongs to~$N$ and is not pumpable with bound~$p-1$ in
$N\subseteq M$. Therefore $p-1$ is not a pumping length for~$N$,
which implies $q\geq p$ (valid pumping lengths are upward closed).

If $q=p$, every $u\in N$ fails to be pumpable in~$M$ with any
bound below~$p$, hence also fails to be pumpable in~$N$ with any
bound below~$q$. Thus $N\subseteq W^{\mpl}(N)$, and the reverse
inclusion follows from contraction.
\end{proof}

The assumption $p<\infty$ in this lemma matters. When $p=\infty$,
Lemma~\ref{lem:mpl-iteration-unpumpable} instead gives immediate
stabilization of the next iterate; no monotonicity claim for the
pumping lengths is needed in that case.

\begin{theorem}[Fixed-point dichotomy]
\label{thm:mpl-fixed-point-dichotomy}
For every language~$L\subseteq\Sigma^*$, we have $\alpha_L\leq\omega$.
Exactly one of the following alternatives occurs.
\begin{enumerate}
\item Some finite iterate is a fixed point. Then $\alpha_L<\omega$.
\item No finite iterate is a fixed point. Then $\alpha_L=\omega$,
      every $p_n=\mpl(L_n)$ is finite, and
      \[
      p_0<p_1<p_2<\cdots.
      \]
      Moreover, no word of~$L_\omega$ admits a pumping decomposition
      in~$L_\omega$.
\end{enumerate}
\end{theorem}

\begin{proof}
The first alternative is immediate. Suppose that no finite iterate
is a fixed point. If $\mpl(L_n)=\infty$ for some~$n$, then
Lemma~\ref{lem:mpl-iteration-unpumpable} would make $L_{n+1}$ a
fixed point. Hence all~$p_n$ are finite.
Lemma~\ref{lem:mpl-growth-or-stabilization} gives $p_{n+1}\geq p_n$,
with equality forcing $L_{n+2}=L_{n+1}$. Thus the pumping lengths
strictly increase and are unbounded.

Put $M=L_\omega$ and let $w\in M$. Choose~$n$ with $p_n>|w|$.
Since $w\in L_{n+1}=W^{\mpl}(L_n)$, the word~$w$ has no pumping
decomposition in~$L_n$ with a prefix bound below~$p_n$.
Any pumping decomposition $w=xyz$ in~$M$, however, would also be
one in~$L_n$, with bound $|xy|\leq|w|<p_n$. This is impossible.
Therefore no word of~$M$ is pumpable in~$M$, and
Lemma~\ref{lem:mpl-iteration-unpumpable} gives $W^{\mpl}(M)=M$.
Thus $L_{\omega+1}=L_\omega$, and the absence of finite
stabilization yields $\alpha_L=\omega$.
\end{proof}

\begin{corollary}[Identification with the $\omega$-limit]
\label{cor:mpl-fixed-point-is-omega-limit}
For every language~$L\subseteq\Sigma^*$,
\[
\Fix_{\mpl}(L)=L_\omega=\bigcap_{n\in\N_0}L_n.
\]
\end{corollary}

\begin{proof}
If $\alpha_L=\omega$, this is the definition. Otherwise the
descending sequence of finite iterates is eventually constant with
value~$L_{\alpha_L}$, so its intersection has that value as well.
\end{proof}

\subsection{Regular initial languages}

Regularity is not needed for stabilization by~$\omega$. By
Theorem~\ref{thm:Reg-L-implies-wit-mpl-L_is-regular}, however,
it ensures that every finite iterate is regular.

\begin{corollary}[Regularity of the iterated fixed point]
\label{cor:mpl-fixed-point-regularity}
Let $L\subseteq\Sigma^*$ be regular.
\begin{enumerate}
\item If $\alpha_L<\omega$, then $\Fix_{\mpl}(L)$ is regular.
\item If $\alpha_L=\omega$, then $\Fix_{\mpl}(L)$ is regular
      if and only if it is finite.
\end{enumerate}
Equivalently, $\Fix_{\mpl}(L)$ is regular if and only if
$\alpha_L<\omega$ or $\Fix_{\mpl}(L)$ is finite.
\end{corollary}

\begin{proof}
A finite iterate of~$L$ is regular by
Theorem~\ref{thm:Reg-L-implies-wit-mpl-L_is-regular}.
If $\alpha_L=\omega$, then no word of~$\Fix_{\mpl}(L)=L_\omega$
is pumpable in that language. An infinite regular language over a
finite alphabet has words of unbounded length, so the pumping lemma
would yield such a decomposition. Thus, in this case, a regular fixed
point must be finite. Conversely, every finite language is regular.
\end{proof}

\begin{remark}
\label{rem:mpl-limit-regularity}
The qualification $\alpha_L=\omega$ in the second assertion is
essential. For example, $L=\{\ta\}^*$ has $\mpl(L)=1$ and is already
a fixed point; its $\omega$-limit is infinite and regular.

Regularity of every finite iterate alone would not imply regularity
of their intersection. For example, the descending regular languages
\[
R_n=\{\ta^i \tb^j\mid i,j\geq n\}
    \cup\{\ta^k \tb^k\mid 0\leq k<n\}
\]
satisfy
\[
\bigcap_{n\in\N_0}R_n=\{\ta^k \tb^k\mid k\in\N_0\},
\]
which is not regular. This example is not asserted to be a witness
iteration. Whether a regular language can have $\alpha_L=\omega$
and an infinite fixed point is a separate question, not resolved here.
\end{remark}

\subsection{Sharpness of the stabilization bound}

We next give a regular language for which every finite step is strict.
We use submonoids of the free monoid~$\Sigma^*$, viewed as languages;
the operator $W^{\mpl}$ acts on their words with the usual
concatenation. A submonoid contains~$\lambda$ and is closed under
concatenation; it is \emph{non-trivial} if it also contains a non-empty word.

\begin{lemma}[Propagation of strictness]
\label{lem:propagation-of-strictness}
Let $M\subseteq\Sigma^*$ be a non-trivial regular submonoid satisfying
\begin{equation}
\label{eq:mpl-monoid-right-cancellation}
uv\in M\ \text{and}\ v\in M\quad\Longrightarrow\quad u\in M
\qquad(u,v\in\Sigma^*).
\end{equation}
Put $N=W^{\mpl}(M)$. Then $N$ is again a regular submonoid
satisfying~\eqref{eq:mpl-monoid-right-cancellation}. If $N\subsetneq M$,
then $N$ is non-trivial and $W^{\mpl}(N)\subsetneq N$.
\end{lemma}

\begin{proof}
Write $q=\mpl(M)$. Every nonempty $w\in N$ has $|w|\geq q$:
otherwise $w=\lambda\cdot w\cdot\lambda$ would be a pumping
decomposition in the submonoid~$M$ with bound $|w|<q$.

We claim that
\begin{equation}
\label{eq:mpl-witness-right-extension}
(N\setminus\{\lambda\})M\subseteq N.
\end{equation}
Let $w\in N\setminus\{\lambda\}$ and $v\in M$.
If $wv\notin N$, a pumping decomposition with bound below~$q$
would lie entirely inside~$w$, because $|w|\geq q$.
Write $w=xyr$, with $|xy|<q$ and
$xy^trv\in M$ for all $t\in\N_0$.
By~\eqref{eq:mpl-monoid-right-cancellation}, $xy^tr\in M$ for all~$t$,
contradicting $w\in N$. This proves
\eqref{eq:mpl-witness-right-extension}.

Since $\lambda\in N\subseteq M$, this inclusion implies that~$N$
is a submonoid. To verify the inherited cancellation property, let
$uv,v\in N$. First $u\in M$ by the property for~$M$.
If $u\notin N$, pumping~$u$ in~$M$ with a bound below~$q$ and
appending~$v\in M$ would show $uv\notin N$, a contradiction.
Regularity follows from
Theorem~\ref{thm:Reg-L-implies-wit-mpl-L_is-regular}.

Now assume $N\subsetneq M$ and choose $r\in M\setminus N$.
Then $q\geq2$, because $q=1$ would give $N=M$.
The failure of the pumping condition at~$q-1$ supplies a word of~$N$
of length at least~$q-1>0$. Choose a shortest non-empty $s\in N$.
By~\eqref{eq:mpl-witness-right-extension}, $sr\in N$.

Any decomposition $sr=xyz$ with $y\neq\lambda$ and $|xy|\leq|s|$
pumps down to a word~$hr$ with $|h|<|s|$.
We show that $hr\notin N$. If $h=\lambda$, this follows from
$r\notin N$. Otherwise, suppose $hr\in N$. Cancellation in~$M$
gives $h\in M$, while the choice of~$s$ gives $h\notin N$.
Thus~$h$ is pumpable in~$M$ with a bound below~$q$; appending~$r$
shows that $hr\notin N$, again a contradiction.

Hence $sr$ is not pumpable in~$N$ with bound~$|s|$, and
\[
\mpl(N)>|s|.
\]
On the other hand, $s^t\in N$ for all $t\in\N_0$, so $s$ is
pumpable in~$N$ with bound~$|s|<\mpl(N)$.
Therefore $s\in N\setminus W^{\mpl}(N)$, proving strictness.
\end{proof}

\begin{theorem}[Strict descent for $(\ta+\tb\tb)^*$]
\label{thm:witness-iteration-strictly-descending}
Let $\Sigma=\{\ta,\tb\}$ and $L_0=(\ta+\tb\tb)^*$, where $+$ denotes union
in the regular expression. Then $L_{n+1}\subsetneq L_n$ for every
$n\in\N_0$.
\end{theorem}

\begin{proof}
The language~$L_0$ is a non-trivial regular submonoid. Membership
means that every maximal block of~$\tb$'s has even length.
If $uv,v\in L_0$, all blocks entirely inside~$u$ have even length.
Its final $\tb$-block also has even length: if it merges with the
initial block of~$v$, subtract the even length of that initial block
from the even length of the merged block. Thus $u\in L_0$, proving
\eqref{eq:mpl-monoid-right-cancellation}.

We have $\mpl(L_0)=2$: the word~$bb$ is not pumpable with bound~$1$,
whereas every word of~$L_0$ of length at least~$2$ begins with either~$\ta$ or~$\tb\tb$, which can be pumped.
Since~$\ta$ is pumpable with bound~$1$, we have
$a\in L_0\setminus L_1$. Repeated application of
Lemma~\ref{lem:propagation-of-strictness} proves the assertion.
\end{proof}

\begin{corollary}[Sharpness of the bound $\omega$]
\label{cor:mpl-omega-bound-sharp}
For $L=(\ta+\tb\tb)^*$,
\[
\alpha_L=\omega,\qquad \Fix_{\mpl}(L)=L_\omega=\{\lambda\}.
\]
Consequently, the bound~$\omega$ is sharp even for regular languages.
\end{corollary}

\begin{proof}
Strict descent and Theorem~\ref{thm:mpl-fixed-point-dichotomy} give
$\alpha_L=\omega$. Every~$L_n$ is a submonoid, so their
intersection~$L_\omega$ is a submonoid and contains~$\lambda$.
If it contained a non-empty word~$w$, then $w^t\in L_\omega$ for
all $t\in\N_0$. Pumping~$w$ as a whole would contradict the
absence of pumping decompositions at a genuine~$\omega$ stage.
Thus $L_\omega=\{\lambda\}$.
\end{proof}

This example also proves that $W^{\mpl}$ is not idempotent, since
$L_2\subsetneq L_1$. Its fixed point is finite; the question of an
infinite fixed point reached only at~$\omega$ from a regular initial
language remains unanswered by these results.

\subsection{A conjectured Fibonacci structure}

Define
\[
u_0=\ta,\qquad u_1=\tb\tb,\qquad u_{i+2}=u_{i+1}u_i
\quad(i\in\N_0).
\]
The first words are
\[
u_0=\ta,\quad u_1=\tb\tb,\quad u_2=\tb\tb\ta,\quad
u_3=\tb\tb\ta\tb\tb,\quad u_4=\tb\tb\ta\tb\tb\tb\tb\ta.
\]
These are finite Fibonacci words encoded by $\ta$ and~$\tb\tb$.
If $F_0=0$, $F_1=1$, and $F_{i+2}=F_{i+1}+F_i$, then
\[
|u_i|=F_{i+2}
\]
by the same recurrence and the initial lengths~$1,2$.

Exact finite-automaton computations give the following initial values.
For each listed value of~$i$, the unique shortest non-empty word is~$u_i$.
\begin{center}
\begin{tabular}{c|rrrrrrrrr}
$i$ & 0 & 1 & 2 & 3 & 4 & 5 & 6 & 7 & 8\\\hline
$\min\{|w|\mid w\in L_i\setminus\{\lambda\}\}$
    & 1 & 2 & 3 & 5 & 8 & 13 & 21 & 34 & 55\\
$\mpl(L_i)$ & 2 & 3 & 5 & 8 & 13 & 21 & 34 & 55 & 89
\end{tabular}
\end{center}
Here pumpability is checked using the full finite orbit of the state
reached after~$x$ under repeated application of~$y$. The continuation~$z$
must lead from every state in this orbit to an accepting state.
Thus all pumping exponents are covered, rather than a prescribed finite
sample. Minimal pumping lengths are found by checking whether the
regular language $L_i\setminus\operatorname{Pump}_r(L_i)$ contains
a word of length at least~$r$. Uniqueness at the shortest non-empty
length is checked by counting accepted words at that length.
These finite computations motivate the following conjecture but do not
prove its assertion for arbitrary~$i$.

\begin{conjecture}[Fibonacci structure of the witness iteration]
\label{conj:mpl-witness-fibonacci}
For the iterates of $L_0=(\ta+\tb\tb)^*$ and every $i\in\N_0$,
the word~$u_i$ is the unique shortest non-empty word of~$L_i$, and
\[
\mpl(L_i)=|u_{i+1}|=F_{i+3}.
\]
Equivalently, the assertion about shortest words is
\[
L_i\cap\Sigma^{<|u_i|}=\{\lambda\},\qquad
L_i\cap\Sigma^{|u_i|}=\{u_i\},
\]
where $\Sigma^{<m}=\{w\in\Sigma^*\mid |w|<m\}$.
\end{conjecture}

The two conjectured assertions together would imply that $u_{i+1}$
is the unique shortest non-empty witness of~$L_i$, of length
$\mpl(L_i)$. They would also imply $u_i\in L_i\setminus L_{i+1}$:
the submonoid~$L_i$ allows pumping~$u_i$ as a whole with bound
$|u_i|<|u_{i+1}|=\mpl(L_i)$.
Thus the conjecture would identify the successive shortest words
removed and determine the exact growth of the pumping lengths.
Neither strict descent nor the computation of the fixed point above
depends on this conjecture.

	\section{Unary Witness Languages}\label{sec:unary-witness}
	As seen in the previous section, the situation for non-unary languages gets more involved. Thus, we continue with the investigation of unary witness languages, i.e., the words that are not $K$-pumpable w.r.t. $K \in \{\mpc, \mpl\}$ for some given unary regular language $L$, in order to learn more about their structure and combinatorial properties.
For example, we observe that the exponential blow-up on the state complexity of witness languages for (non-unary) regular languages we discussed in \Cref{thm:exponential-size-binary} is not happening in unary languages.  For unary languages $L$ we have $|M(L)| = \dsc(L)$ where $M(L)$ is the syntactic monoid.

\ifpaper
\else
The following lemma is proven in~\cite{GrHoRa24c} and will be useful for the study of witness languages.

\begin{lemma}\label{lem:a-greater-word-is-missing}
	Let~$L$ be a unary regular language 
	over~$\Sigma=\{\ta\}$. Then $\mpc(L)>\ell$ if there is a
	word~$w=\ta^\ell$ in~$L$ such that
	\begin{enumerate}
		\item there is \emph{no} word~$\ta^k\in L$, for $0\leq k<\ell$, or
		\item  the property $\ta^k(\ta^{\ell-k})^*\not\subseteq L$ holds for every
		word~$\ta^{k}\in L$ with~$k<\ell$.
	\end{enumerate}
	%
\end{lemma}

Indeed, the reader may verify that the word~$\ta^\ell$ mentioned in the previous lemma is a witness. 
\fi

We already classified the set of witness languages w.r.t.~$\mpc$ in the non-unary and the unary case (Proposition~\ref{prop:every-set-is-witnesslanguage}).
One may wonder to what extent we can extend \Cref{lem:mpl-finite-languages-bounded} in the unary case.
Indeed, there are languages whose pumping constant $\mpl$ is arbitrarily smaller than the according pumping constant of its witness language as shown in Lemma~\ref{lem:mpl-distance-arbitrarily-large}.
Another attempt to extend the statement of Lemma~\ref{lem:mpl-finite-languages-bounded} is to encompass all finite languages. But, as we will see next, this is not possible.

\ifpaper
\begin{lemma}\label{lem:finite-notsubset-witnesslang}
	For $p \in \N$, we have~$\mathcal{FIN} \not\subset \mathcal{W}_p^{\mpl,u}.$
\end{lemma}
\begin{proof}
	For~$p=1$ we observe that~$\ta^*$ is the only unary language~$L$ fulfilling~$\mpl(L)=1$. We observe that~$\lambda$ is the only witness of~$L$ which implies that~$\mathcal{W}_1^{\mpl,u}=\{\{\lambda\}\}$.
	
	Hence, we can assume for the rest of the proof that~$p\ge 2$.
	Additionally, assume there is a language~$L$ with~$\mpl(L)=p$ and~$W(L)=\ta^{\le x}$ for~$x\ge 2p$. As a direct consequence of Lemma~\ref{lem:mpl} we know that each word in~$\ta^{\le 2p}\setminus \ta^{\le p}$ can be pumped. On the other hand, each of these words is in~$W(L)$ which implies that they can only be pumped by the factor~$\ta^p$. Hence, we have
	$$
	\bigcup_{i=p+1}^{2p}\ta^i\cdot (\ta^p)^*=\ta^{\ge p+1}
	$$
	is a subset of the language~$L$.
	Since~$W(L)\subseteq L$ we obtain that~$L$ contains
	$$
	W(L)\cup \ta^{\ge p+1}=\ta^{\le x}\cup \ta^{\ge p+1}=\ta^*,
	$$
	which implies that~$L=\ta^*$. In turn, this contradicts the assumption that~$\mpl(L)=p>1$ and therefore proves the stated claim.\iflipics\else\qed\fi
\end{proof}
\else
\begin{lemma}\label{lem:finite-notsubset-witnesslang}
	For $p \in \N$, we have~$\mathcal{FIN} \not\subset \mathcal{W}_p^{\mpl,u}.$
\end{lemma}
\begin{proof}
	For~$p=1$ we observe that~$\ta^*$ is the only unary language~$L$ fulfilling~$\mpl(L)=1$. We observe that~$\lambda$ is the only witness of~$L$ which implies that~$\mathcal{W}_1^{\mpl,u}=\{\{\lambda\}\}$.
	
	Hence, we can assume for the rest of the proof that~$p\ge 2$.
	Additionally, assume there is a language~$L$ with~$\mpl(L)=p$ and~$W(L)=\ta^{\le x}$ for~$x\ge 2p$. As a direct consequence of Lemma~\ref{lem:mpl} we know that each word in~$\ta^{\le 2p}\setminus \ta^{\le p}$ can be pumped. On the other hand, each of these words is in~$W(L)$ which implies that they can only be pumped by the factor~$\ta^p$. Hence, we have
	$$
	\bigcup_{i=p+1}^{2p}\ta^i\cdot (\ta^p)^*=\ta^{\ge p+1}
	$$
	is a subset of the language~$L$.
	Since~$W(L)\subseteq L$ we obtain that~$L$ contains
	$$
	W(L)\cup \ta^{\ge p+1}=\ta^{\le x}\cup \ta^{\ge p+1}=\ta^*,
	$$
	which implies that~$L=\ta^*$. In turn, this contradicts the assumption that~$\mpl(L)=p>1$ and therefore proves the stated claim.\iflipics\else\qed\fi
\end{proof}
\fi

Next, we want to study the relation of the scattered factor distance of a witness language and the pumping constant $\mpl$ of the according language.
Therefore, we observe the following useful characterisation of unary languages.

\ifpaper
\begin{lemma}\label{lem:index-collection}
	Let $p \in \N$. For each~$L\subseteq\Sigma^*$ with~$\mpl(L)=p$ there is a finite index set~$I_L\subseteq \N\times \{2, 3,\ldots, p\}$ and two finite sets~$S\subseteq \Sigma^{\le p-1}$, $S'\subseteq \Sigma^*$ such that
	$$
	L=\bigcup_{(i,j)\in I_L} \ta^i(\ta^j)^*\cup S
	\quad \text{and} \quad
	W^{\mpl}(L)\subseteq \bigcup_{(i,j)\in \widehat{I_L}} \ta^i(\ta^j)^*\cup S',
	$$
	where~$\widehat{I_L}\subset \N\times \{p\}$.
\end{lemma}
\begin{proof}
	Given any word~$w\in L$. If~$w$ is not pumpable w.r.t. any value less or equal to~$p$, then~$w$ has to be shorter than~$p$.
	On the other hand, if~$w$ is pumpable w.r.t. the value~$j$ then~$w=\ta^{|w|\bmod j}(\ta^j)^k$ for a natural number~$k>1$.
	Indeed, this implies that~$\ta^{|w|\bmod j+(k-1)j}(\ta^j)^*$ is a subset of~$L$.
	By repeating this procedure for all words in~$L$ in increasing lexicographic order\footnote{This algorithm terminates because~$L$ is a regular language.} we obtain the set~$I_L$ and therefore the first claim follows.
	The second claim is an implication of the fact, that each witness cannot be pumped by any value less than~$p$.
	\iflipics\else\qed\fi\end{proof}
\else
\begin{lemma}\label{lem:index-collection}
	Let $p \in \N$. For each~$L\subseteq\Sigma^*$ with~$\mpl(L)=p$ there is a finite index set~$I_L\subseteq \N\times \{2, 3,\ldots, p\}$ and two finite sets~$S\subseteq \Sigma^{\le p-1}$, $S'\subseteq \Sigma^*$ such that
	$$
	L=\bigcup_{(i,j)\in I_L} \ta^i(\ta^j)^*\cup S
	\quad \text{and} \quad
	W^{\mpl}(L)\subseteq \bigcup_{(i,j)\in \widehat{I_L}} \ta^i(\ta^j)^*\cup S',
	$$
	where~$\widehat{I_L}\subset \N\times \{p\}$.
\end{lemma}
\begin{proof}
	Given any word~$w\in L$. If~$w$ is not pumpable w.r.t. any value less or equal to~$p$, then~$w$ has to be shorter than~$p$.
	On the other hand, if~$w$ is pumpable w.r.t. the value~$j$ then~$w=\ta^{|w|\bmod j}(\ta^j)^k$ for a natural number~$k>1$.
	Indeed, this implies that~$\ta^{|w|\bmod j+(k-1)j}(\ta^j)^*$ is a subset of~$L$.
	By repeating this procedure for all words in~$L$ in increasing lexicographic order\footnote{This algorithm terminates because~$L$ is a regular language.} we obtain the set~$I_L$ and therefore the first claim follows.
	The second claim is an implication of the fact, that each witness cannot be pumped by any value less than~$p$.
	\iflipics\else\qed\fi\end{proof}
\fi

%
%
%
%


By using the previous lemma we observe that despite Corollary~\ref{cor:mpl-sd-infty} the scattered factor distance of a witness languages is exponentially bounded by~$\mpl$.

\ifpaper
\begin{reptheorem}{thm:sd-W-distance-not-arbitrarily-large}
Let $p\in \N$ and~$L$ be a unary regular language with~$\mpl(L)=p\ge 2$. 
Then we have
$$\sd(W^{\mpl}(L))\le \lcm(2,3,\ldots,p-1)\cdot p= e^{(1+o(1))p}.$$
\vspace{-6mm}
\end{reptheorem}
\noindent \textit{Proof idea.}
First, we observe that each unary word can be uniquely identified by its length which can be written as~$px+i$ for two numbers~$x,i \in \N$. Hence, natural numbers and words are used interchangeably throughout this proof. The scattered factor distance of a unary language is thus equal to the largest difference of numbers~$px+i$ and~$py+j$ such that the according words are in the language.

Let~$L$ be a unary regular language with~$\mpl(L)=p$.
Using Lemma~\ref{lem:index-collection}, we obtain a representation~$L=\bigcup_{(i,j)\in I_L\atop x\in \N_0} i+jx\cup S$ for some finite set $S \subseteq \Sigma^{\leq p-1}$.
Indeed, a word~$px+i$ is in~$L$ if $(x,i)\in I_L$ or the words~$px+i\in S$.

For showing that $\sd(W^{\mpl}(L))<\ell$, we have to prove that for each $\ell+1$ consecutive numbers at least two of them refer to words in~$W^{\mpl}(L)$.

Indeed, a word~$px+i$ is a witness of~$L$ if it is not pumpable w.r.t. any value less than~$p$.
W.l.o.g. we assume that~$px+i\notin S$.
Therefore, $px+i$ cannot be equal to any number $i+jz$ which fulfils~$(i,j)\in I_L$ and~$z\in\N_0$. 
We observe that the tuple~$(i,j)$ in~$I_L$ refers to the residue class modulo~$j$ with remainder~$i$.
Since the index set~$\widehat{I_L}$ of the witness language~$W^{\mpl}(L)$ contains only tuples of the form~$(i,p)$, the value of~$\sd(W^{\mpl}(L))$ is bound from above by~$p\cdot \Theta$ such that for~$0\le k\le x-1$ there is~$(i',j')\in I_L$ with~$j'<p$ and~$i+pk\bmod j=i'$.
By assuming~$I_L$ is maximal in the sense that, for each~$j<p$ there are exactly~$j-1$ residue classes modulo~$j$ induced by the tuples in~$I_L$, we have that~$\Theta$ is a solution of the system of linear congruences
\begin{align*}
	i+px &\equiv k_2 \bmod 2,& \quad
	i+px &\equiv k_3 \bmod 3,& \quad
	i+px &\equiv k_4 \bmod 4,&\\
	& \dots  \hspace{0.5cm}\dots&\quad
	i+px &\equiv k_{p-2} \bmod p-2,&\quad
	i+px &\equiv k_{p-1} \bmod p-1,&
\end{align*}
for some remainders~$k_2, k_3, \ldots, k_{p-1}$.
We obtain that~$\Theta\le \lcm(2,3,\ldots, p-1)$ because~$\Theta\bmod \lcm(2,3,\ldots, p-1)$ is also a solution of the system.
Therefore, we have
$$\sd(W^{\mpl}(L))\le (\Theta+1)\cdot p\le (\lcm(2,3,\ldots,p-1)-1+1)\cdot p= \lcm(2,3,\ldots,p-1)\cdot p.$$

An equivalent formulation of the prime number theorem states that
\[
\log\lcm(1,2,\ldots,n)\sim n
\qquad\text{as }n\to\infty;
\]
see, for instance, Hardy and Wright~\cite{HaWr79}, which implies the stated inequivalence
\[
\sd\bigl(W^{\mpl}(L)\bigr)\leq \lcm(2,3,\ldots,p-1)\cdot p=e^{(1+o(1))p}.\hfill\qed
\]

\else
\begin{reptheorem}{thm:sd-W-distance-not-arbitrarily-large}
Let $p\in \N$ and~$L$ be a unary regular language with~$\mpl(L)=p\ge 2$. 
Then we have
$$\sd(W^{\mpl}(L))\le \lcm(2,3,\ldots,p-1)\cdot p= e^{(1+o(1))p}.$$
\vspace{-6mm}
\end{reptheorem}
\begin{proof}
	In order to show the claim, we first inspect the words in the witness language~$W^{\mpl}(L)$ and show that their exponents will be of the form $px +i$ and $py + j$ for some multiples $x,y \in \N$ and offsets $i,j \in \N$, then we observe that the maximal
	difference between these exponents can be derived from the elements in~$I_L$.
	
	Let~$I_L\subseteq \N\times\{2,3,\ldots,p\}$ and~$\widehat{I_L}\subseteq  \N\times\{p\}$
	be index sets for~$L$ and~$W^{\mpl}(L)$, \resp, according to Lemma~\ref{lem:index-collection}.
	Without loss of generality we assume that~$I_L$ and~$S$ are the smallest sets describing~$L$.
	For every number~$j$ there are exactly~$j$ residue class modulo~$j$ and the languages~$\ta^i(\ta^j)^*$ can be interpreted as the residue classes modulo~$j$ with the remainder~$i$.
	Since these languages have a one-to-one correspondence with the tuples~$(i,j)$ we observe that the set~$I_L$ contains at most~$j$ pairs~$(i,j)$ for each number~$2\le j\le p$.
	If~$I_L$ contains exactly~$j$ pairs for a number~$2\le j\le p$ then we have~$\ta^{\ge i_j}\subseteq L$ for~$i_j=\max\{i\mid (i,j)\in I_L\}$.	Hence, 
	each word longer than~$i_j$ can be pumped by its first letter and thus they are not in~$W^{\mpl}(L)$. In turn, this implies that $$\sd(W^{\mpl}(L))\le\sd(W^{\mpl}(\bigcup_{(i,j)\in I_L\atop i<  i_j} \ta^i(\ta^j)^*\cup S)).$$
	Therefore, we can assume that~$|I_L\cap \{j\}|\le j-1$ for each number~$2\le j\le p$.
	Since we know that	
	$$W^{\mpl}(L)\subseteq \bigcup_{(i,j)\in \widehat{I_L}} \ta^i(\ta^j)^*\cup S',
	$$
	with~$\widehat{I_L}\subset \N\times \{p\}$ we derive that~$\sd(W^{\mpl}(L))\le\max\{\sd(S'),p\}$.
	Hence, it suffices to show that~$\sd(S')$ is bounded from above by some~$\Theta \in \N$.
	Therefore, we observe that
	$$W^{\mpl}(L)\subseteq \bigcup_{(i,p)\in I_L\atop i\ge p} \ta^{i \bmod p}(\ta^p)^*\cup (S'\cap \ta^{\le p-1}),
	$$
	wherefore this task can be reformulated to finding the largest range~$[k,k']$
	such that each word~$w$ in
	$$\bigcup_{(i,p)\in I_L\atop i\ge p} \ta^{i\bmod p}(\ta^p)^*\cap \{\ta^\ell\mid k\le \ell\le k'\}$$ is also in
	$$\bigcup_{(i,j)\in I_L\atop j\ne p} \ta^i(\ta^j)^*.$$
	For the sake of simplicity, we increase the range~$[k,k']$ further and thus give a larger bound for~$\sd$ by just inspecting the words in~$\ta^i(\ta^p)^*\cap \{\ta^\ell\mid k\le \ell\le k'\}$ for some fixed value~$i\ge p$.
	Thus, the task is equivalent to including the largest possible number of consecutive multiples of~$p$ plus~$i$ into at least one of the residue classes which are represented by the remainder~$i$ modulo~$j$ for~$(i,j)\in I_L$ and~$j\ne p$.
	
	Again for the sake of simplicity, we assume that all residue classes modulo~$j$ are in~$I_L$ except one, i.e., we have~$(i,j)\in I_L$ for all~$0\le i\le j-1$ where~$i\ne k_i$ for an integer~$0\le k_i\le j-1$.
	Hence, we obtain that the numbers~$i+px$ are equivalent to~$k_i$ modulo~$j$ for all words~$\ta^{i+px}\in W^{\mpl}(L)$.
	By applying this argument to all numbers~$2\le j\le p-1$ we obtain that all words~$\ta^{i+px}\in W^{\mpl}(L)$ fulfil
	\begin{align*}
		i+px &\equiv k_2 \bmod 2,&\\
		i+px &\equiv k_3 \bmod 3,&\\
		i+px &\equiv k_4 \bmod 4,&\\
		& \vdots  \hspace{1.6cm}\vdots&\\
		i+px &\equiv k_{p-2} \bmod p-2,&\\
		i+px &\equiv k_{p-1} \bmod p-1.&
	\end{align*}
	Since we have~$\sd(W^{\mpl}(L))=|w|-|w'|$ for two words~$w,w'\in W^{\mpl}(L)$
	we can alternatively assume that~$\Theta=|w|-|w'|\bmod p$ fulfils the previous system of linear congruences.
	Thus, we have~$|[k,k']|=\Theta\cdot p$.
	Indeed, with the previous assumption on~$\Theta$ we have implicitly assumed that~$w,w'$ are not in~$S$.
	By allowing~$|w'|=0$ and setting~$\lambda\in S$ which implies that~$\lambda\in S'$,
	we can increase~$\sd(W^{\mpl}(L))$ by at most~$p$.
	Hence, we have~$\sd(W^{\mpl}(L))\le |[k,k']|+p=\Theta\cdot p+p=(\Theta+1)\cdot p$.
	
	It remains to show that~$\Theta \le \lcm(2,3,\ldots, p-1)$.
	Suppose the previously shown system of linear congruences is solvable, and let $\Theta$ and $\Theta'$ be two
	solutions. Then for each $2\le i\le p-1$ we have $\Theta \equiv k_i \equiv \Theta' \pmod{i}$ and thus $\Theta\equiv\Theta'\pmod{i}$.
	Indeed this implies that~$i \mid (\Theta-\Theta')$ for all $2\le i\le p-1$.
	Therefore, we obtain that~$\lcm(2,3,\ldots,p-1)\mid (\Theta-\Theta')$. In other words all solutions of the system have distance $\lcm(2,3,\ldots,p-1)$.
	Let~$\Theta = r \mod +\lcm(2,3,\ldots,p-1) s$ for integers~$r$ and $s$ with~$0\le r\le \lcm(2,3,\ldots,p-1)-1$.
	We see that
	$$k_i\equiv \Theta \equiv  r +\lcm(2,3,\ldots,p-1) s \equiv r\mod i,$$ because $i$ is a divisor of $\lcm(2,3,\ldots,p-1)$.
	Therefore, we obtain that~$r$ is a solution of the aforementioned system of linear congruences. This implies that the smallest nonnegative solution is bound from above by~$\lcm(2,3,\ldots,p-1)-1$ and thus
	$$\sd(W^{\mpl}(L))\le (\Theta+1)\cdot p\le (\lcm(2,3,\ldots,p-1)-1+1)\cdot p= \lcm(2,3,\ldots,p-1)\cdot p.$$
	
	An equivalent formulation of the prime number theorem states that
	\[
	\log\lcm(1,2,\ldots,n)\sim n
	\qquad\text{as }n\to\infty;
	\]
	see, for instance, Hardy and Wright~\cite{HaWr79}.
	Therefore, it follows that
	\[
	\log\lcm(2,3,\ldots,p-1)
	=
	(1+o(1))p.
	\]
	Moreover, $\log p=o(p)$, and hence
	\[
	\begin{aligned}
		\log\bigl(
		p\cdot\lcm(2,3,\ldots,p-1)
		\bigr)
		&=
		\log p+\log\lcm(2,3,\ldots,p-1)\
		&=
		(1+o(1))p.
	\end{aligned}
	\]
	Exponentiating this estimate yields
	\[
	\sd\bigl(W^{\mpl}(L)\bigr)\leq \lcm(2,3,\ldots,p-1)\cdot p=e^{(1+o(1))p}.
	\]
	Thus, the scattered-factor distance of the witness language is bounded from above by a quantity that is exponential in the minimal pumping
	length~$p$.
	\iflipics\else\qed\fi
\end{proof}

\fi

	\section{The Witness-Language Problem}\label{sec:witness-problem}
	
This section is dedicated to the computational complexity of the \textsc{Pumping-Witness-Language-Problem} or, for short, \textsc{Witness-Language-Problem} which is defined in the following way:


\begin{description}
	\item[]\textsc{Pumping-Witness-Language-Problem}:
	\begin{description}
		\item[]\textsc{Fixed}: Particular pumping lemma, such as
		Lemma~\ref{lem:mpc}.
		\item[]\textsc{Input}: Two accepting or generating devices~$A,B$ for a
		formal language family such as, e.g., the regular languages.
		\item[]\textsc{Output:} Yes, if and only if~$W(L(A))=L(B)$
		w.r.t.\ previously fixed pumping lemma.
	\end{description}
\end{description}


The following auxiliary result provides an oracle for testing whether every replication of a word is in a given language.

\begin{lemma}[\cite{GHR23}]\label{lem:bounded-inclusion}
	Given a \dfa\ $A=(Q,\Sigma, \cdot, q_0, F)$ and a word~$w\in \Sigma^*$, the language inclusion problem for~$w^*$ in~$L(A)$ is
	\L-complete.
\end{lemma}

First, we see that the \textsc{Witness-Language-Problem} is contained in \co\NP.

\ifpaper
\begin{theorem}\label{thm:reg-reg-mpc-conp-containment}
	Given two \dfa s~$A$ and~$B$, then, it can be decided in \co\NP\
	whether~$W(L(A))=L(B)$ holds w.r.t.\ Lemma~\ref{lem:mpc}.
\end{theorem}

\begin{proof}
	The proof follows a similar argumentation as presented in \cite{GHR23} for the \textsc{Pumping-Problem}.
	
	Let $\langle A, B\rangle$ be an input instance of the problem in question.
	We construct a $\co\NP$ Turing machine $M$ with access to an~$\L$-oracle: 
	first the computation universally guesses ($\forall$-states) a word~$w$ with
	$1\le |w| \le \max\{\dsc(A), \dsc(B)\}$. On that particular branch~$M$ checks deterministically if~$w\in L(A)$ and if~$w\in L(B)$.
	If~$w\in L(B)\setminus L(A)$ then~$M$ stops and rejects.
	In the case that~$w\notin L(B)\cup L(A)$ the machine~$M$ stops and accepts.
	For the remaining case, namely~$w\in L(A)$, $M$ cycles deterministically through all valid decompositions~$w=xyz$ with~$y\ne \lambda$. 
	Then it constructs a \dfa~$C$ accepting the language quotient
	$$(x^{-1}L(A))z^{-1}.$$
	Afterwards the Turing machine~$M$ decides whether~$y^*\subseteq L(C)$ with the 
	\co\NP-oracle provided by Lemma~\ref{lem:bounded-inclusion}.
	If~$y^*\subseteq L(C)$ we obtain that~$w=xyz$ can be pumped by~$y$ and thus~$w$ is not contained in~$W(L(A))$.
	Then the Turing machine~$M$ stops and rejects if~$w\in L(B)$.
	Otherwise,~$M$ continues with the next decomposition of~$w$ until either the previously described case occurs or no other decomposition of~$w$ is left.
	In the latter case, the word~$w\in L(A)$ is not pumpable by any of its decompositions. Hence, we have~$w\in W(L(A))$.
	Then the Turing machine accepts if~$w\in L(B)$ and rejects otherwise.
	\iflipics\else\qed\fi
\end{proof}
\else
\begin{theorem}\label{thm:reg-reg-mpc-conp-containment}
	Given two \dfa s~$A$ and~$B$, then, it can be decided in \co\NP\
	whether~$W(L(A))=L(B)$ holds w.r.t.\ Lemma~\ref{lem:mpc}.
\end{theorem}

\fi

Although the \textsc{Witness-Language-Problem} is related to the \textsc{Pumping-Problem}---which asks for the smallest value~$p$ satisfying a particular pumping lemma for~$L(A)$---it is not possible to use the construction of its \co\NP-hardness as shown in~\cite[Theorem 14]{GHR23} for the \textsc{Witness-Language-Problem}, 
because it leads to a witness language accepted by minimal \dfa\ that probably has an exponential number of states and thus is not usable as an input device.
\ifpaper
Nevertheless, it is possible to reduce the problem whether a given formula in 3DNF is a tautology to the \textsc{Witness-Language-Problem}.
Due to space constraints, we present the proof idea while the complete reduction along with an exemplary construction of an automaton $A_\varphi$ for a given formula $\varphi$ in 3DNF
can be found in the long version of this paper.
\else
Nevertheless, we present a reduction from the problem whether a given formula in 3DNF is a tautology to the \textsc{Witness-Language-Problem}.
An exemplary construction of an automaton $A_\varphi$ for a given formula $\varphi$ in 3DNF is given in \Cref{fig:conp-hardness-automaton-A}.
\fi

\ifpaper
\begin{reptheorem}{thm:reg-reg-mpc-conp-hard}
	Let~$A$ and~$B$ be \dfa s. Then it is \co\NP-hard to decide
	whether $W^{\mpc}(L(A))=L(B)$ holds.
\end{reptheorem}
\noindent \textit{Proof idea.}
	In order to show the \co\NP-hardness of the \textsc{Witness-Language-Problem} we construct for a given formula~$\varphi$ in 3DNF with~$n$ variables and~$m$ conjunctions a \dfa\ $A_\varphi$ which fulfils~$W^{\mpc}(L(A_\varphi))=\{\lambda\}$ if and only if~$\varphi$ is a tautology. For technical reasons one needs to inspect $W^{\mpc}(L(A_\varphi) \cap \Sigma^{\neq 0 \bmod n})$ instead of $W^{\mpc}(L(A_\varphi))$ which we will omit in this proof idea.
	
	Assume that $\varphi=\bigvee_{i=1}^m C_i$ is a formula in 3DNF with~$n$
	variables $x_1,x_2,\ldots,x_n$ and~$m$ conjunctions.
	We define a \dfa~$A_{\varphi}$ which consists of~$n$ states followed by~$m$ blocks of states and an accepting sink state; one block for each conjunction. The block for conjunction~$C_i$ consists of~$n\cdot (|C_i|+1)$ states and can be seen as a table with~$n$ rows and~$|C_i|+1$ columns, where each row represents one of the variables of~$\varphi$.
	Since~$\varphi$ is in 3DNF we know that~$|C_i|\le 3$.
	Each of the columns can be seen as one of the integers from~$-1$ to $|C_i|-1$ 
	representing that the according number of literals has been evaluated to \textit{true} beforehand where the number~$-1$ represents that one of the literals of~$C_i$ has been evaluated to \textit{false} beforehand. 
	In addition, each word~$w\in(\{\mathtt{0},\mathtt{1}\}^n)^*$ represents an assignment for the variables of~$\varphi$, where the~\nth{$(n\cdot i+j)$} letter of~$w$ is the assignment of~$x_j$ in~$C_i$.
	
	The automaton~$A_{\varphi}$ works as follows.
	It starts its computation in an accepting state which is mapped by each word of length~$n$ onto an accepting state again, while each non-empty shorter word is not accepted.
	Afterwards~$A_{\varphi}$ counts for each conjunction how many literals of this conjunction are evaluated to \textit{true} by moving from a state according to variable~$x_j$ for~$1\le j\le n-1$ in column~$k$ to the state according to~$x_{j+1}$ in \textit{(i)} column~$k+1$ for~$0\le k\le |C_i|-1$ if~$x_{j+1}\in C_i$ and the input is~$\mathtt{1}$ or~$\overline{x}_j\in C_i$ and the input is~$\mathtt{0}$; \textit{(ii)} column~$-1$ for~$0\le k\le |C_i|-1$ if~$x_{j+1}\in C_i$ and the input is~$\mathtt{0}$ or~$\overline{x}_j\in C_i$ and the input is~$\mathtt{1}$; \textit{(iii)} to the accepting sink-state for~$k=|C_i|$ if~$x_{j+1}\in C_i$ and the input is~$\mathtt{1}$ or~$\overline{x}_j\in C_i$ and the input is~$\mathtt{0}$; \textit{(iv)} in column~$k$ if neither~$x_{j+1}$ nor~$\overline{x}_j$ are in~$C_i$.
	The last states for each conjunction (except the last conjunction) are accepting and thus enable~$A_{\varphi}$ to accept the words~$(\Sigma^n)^{\le m}$.
	If the automaton switches from one conjunction to the next one it resets its computation, i.e., the last state of each conjunction which are not mapped onto the accepting sink state are treated as if they were in column~$0$.
	The last states according to the last conjunction are non-accepting sink states which forces~$A_{\varphi}$ to reject every word in~$(\Sigma^n)^{>m}$.
	If a word~$w$ maps the start state of~$A_{\varphi}$ onto the accepting sink state it can be pumped by its last~$n+(|w|\bmod n)$ letters.
	Otherwise, the word cannot be pumped because~$(\Sigma^n)^{>m}\not\subseteq L(A_{\varphi})$ and thus is a witness of~$L(A_{\varphi})$. In summary, if each word can be pumped then none of the words (except $\lambda$) is a witness of~$L(A_{\varphi})$ and otherwise~$W^{\mpc}(L(A_{\varphi}))\ne \{\lambda\}$.
	Since the first case appears if and only if~$\varphi$ is a tautology,
	the previously shown construction\footnote{Since we have~$\mpc(L({A_\varphi}))=1$ the construction re-establishes the \co\NP-hardness of the \textsc{Pumping-Problem}.} is a polytime reduction from \textsc{Tautology} onto the \textsc{Witness-Language-Problem}.
	Thus, the \textsc{Witness-Language-Problem} for \dfa s w.r.t. Lemma~\ref{lem:mpc} is \co\NP-hard.
	The automaton~${A_\varphi}$ is shown schematically in Figure~\ref{fig:conp-hardness-automaton-A}, where~$\varphi=(x_1\wedge \bar{x}_2\wedge x_4)\vee (\overline{x}_1\wedge x_3)$.\qed
	
	\begin{figure}[!htb]
		\centering
		\begin{tikzpicture}[->,>=stealth',shorten >=1pt,auto,node distance=2.9cm,
	semithick,scale=0.7, every node/.style={scale=0.7},initial text=]
	\tikzstyle{every state}=[fill={rgb:black,1;white,10},ellipse]

	\node[state,initial,accepting] (q0) at (7,20) {$q_0$};	
	\node[state] (q1) at (7,18.5)  {$q_1$};
	\node[state] (q2) at (7,17)  {$q_2$};	
	\node[state] (q3) at (7,15.5)  {$q_3$};	
	\node[state,accepting] (q4) at (7,14)  {$q_4$};	
	
	\node[state] (q11-1) at (2.5,12)  {$q_{1,1,-1}$};	
	\node[state] (q110) at (5.5,12)  {$q_{1,1,0}$};	
	\node[state] (q111) at (8.5,12)  {$q_{1,1,1}$};	
	\node[state] (q112) at (11.5,12)  {$q_{1,1,2}$};	
	
    \node[state] (q12-1) at (2.5,10)  {$q_{1,2,-1}$};	
	\node[state] (q120) at (5.5,10)  {$q_{1,2,0}$};	
	\node[state] (q121) at (8.5,10)  {$q_{1,2,1}$};	
	\node[state] (q122) at (11.5,10)  {$q_{1,2,2}$};	
	
	\node[state] (q13-1) at (2.5,8)  {$q_{1,3,-1}$};	
	\node[state] (q130) at (5.5,8)  {$q_{1,3,0}$};	
	\node[state] (q131) at (8.5,8)  {$q_{1,3,1}$};	
	\node[state] (q132) at (11.5,8)  {$q_{1,3,2}$};	
	
    \node[state,accepting] (q14-1) at (2.5,6)  {$q_{1,4,-1}$};	
	\node[state,accepting] (q140) at (5.5,6)  {$q_{1,4,0}$};	
	\node[state,accepting] (q141) at (8.5,6)  {$q_{1,4,1}$};	
	\node[state,accepting] (q142) at (11.5,6)  {$q_{1,4,2}$};	

    \node[state] (q21-1) at (2.5,4)  {$q_{2,1,-1}$};	
	\node[state] (q210) at (5.5,4)  {$q_{2,1,0}$};	
	\node[state] (q211) at (8.5,4)  {$q_{2,1,1}$};	
	\node[state] (q212) at (11.5,4)  {$q_{2,1,2}$};	
	
	\node[state] (q22-1) at (2.5,2)  {$q_{2,2,-1}$};	
	\node[state] (q220) at (5.5,2)  {$q_{2,2,0}$};	
	\node[state] (q221) at (8.5,2)  {$q_{2,2,1}$};	
	\node[state] (q222) at (11.5,2)  {$q_{2,2,2}$};	
	
    \node[state] (q23-1) at (2.5,0)  {$q_{2,3,-1}$};	
	\node[state] (q230) at (5.5,0)  {$q_{2,3,0}$};	
	\node[state] (q231) at (8.5,0)  {$q_{2,3,1}$};	
	\node[state] (q232) at (11.5,0)  {$q_{2,3,2}$};	
	
    \node[state] (q24-1) at (2.5,-2)  {$q_{2,4,-1}$};	
	\node[state] (q240) at (5.5,-2)  {$q_{2,4,0}$};	
	\node[state] (q241) at (8.5,-2)  {$q_{2,4,1}$};	
	\node[state] (q242) at (11.5,-2)  {$q_{2,4,2}$};	
	
	\node[state,accepting] (q5) at (14,5){$q_5$};
	
	\node[] (qa) at (-1,13) {$x_1\in C_1$};
	
	\node[] (qb) at (-1,11) {$\overline{x}_2\in C_1$};
	
	\node[] (qc) at (-1,9) {$x_3,\overline{x}_3\not\in C_1$};
	
	\node[] (qd) at (-1,7) {$x_4\in C_1$};
	
	\node[] (qe) at (-1,5) {$\overline{x}_1\in C_2$};
	
	\node[] (qf) at (-1,3) {$x_2,\overline{x}_2\not\in C_2$};
	
	\node[] (qg) at (-1,1) {$x_3\in C_2$};
	
	\node[] (qh) at (-1,-1) {$x_4,\overline{x}_4\not\in C_2$};
	
\draw   
    (q0) edge[ultra thick, ->,left] node{$\boldsymbol{0},\mathtt{1}$}       (q1)
    
    (q1) edge[->,left,ultra thick] node{$\boldsymbol{0},\mathtt{1}$}       (q2)
            
    (q2) edge[->, left,ultra thick] node{$\boldsymbol{0},\mathtt{1}$}     (q3)
    
    (q3) edge[->, left,ultra thick] node{$\boldsymbol{0},\mathtt{1}$}  (q4)
            
    (q4) edge[->, above,ultra thick] node{$\boldsymbol{0}$}  (q11-1)
    (q4) edge[->, above] node{$\mathtt{1}$}  (q111)
    
    (q11-1) edge[->, left,ultra thick] node{$\boldsymbol{0},\mathtt{1}$}  (q12-1)
    
    (q110) edge[->, above] node{$\mathtt{0}$}  (q121)
    (q110) edge[->, above] node{$\mathtt{1}$}  (q12-1)
    
    (q111) edge[->, left,above] node{$\mathtt{0}$}  (q122)
    (q111) edge[->, above] node{$\mathtt{1}$}  (q12-1)
    
    (q112) edge[->, bend left, right] node{$\mathtt{0}$}  (q5)
    (q112) edge[->, near start,below] node{$\mathtt{1}$}  (q12-1)
    
    (q12-1) edge[->, left,ultra thick] node{$\boldsymbol{0},\mathtt{1}$}  (q13-1)
    
    (q120) edge[->, left] node{$\mathtt{0},\mathtt{1}$}  (q130)
    
    (q121) edge[->, left] node{$\mathtt{0},\mathtt{1}$}  (q131)
    
    (q122) edge[->, left] node{$\mathtt{0},\mathtt{1}$}  (q132)
    
    (q13-1) edge[->, left,ultra thick] node{$\boldsymbol{0},\mathtt{1}$}  (q14-1)
    
    (q130) edge[->, above] node{$\mathtt{0}$}  (q14-1)
    (q130) edge[->, above] node{$\mathtt{1}$}  (q141)
    
    (q131) edge[->, above] node{$\mathtt{0}$}  (q14-1)
    (q131) edge[->, left] node{$\mathtt{1}$}  (q142)
    
    (q132) edge[->, below, near start] node{$\mathtt{0}$}  (q14-1)
    (q132) edge[->, left] node{$\mathtt{1}$}  (q5)
    
    (q14-1) edge[->,above, ultra thick] node[name=q14-10]{}  (q211)

    (q14-1) edge[->, left] node{$\mathtt{1}$}  (q21-1)
       
    (q140) edge[->, above,near end] node{$\mathtt{0}$}  (q211)
    (q140) edge[->, above, near end] node{$\mathtt{1}$} (q21-1)
       
    (q141) edge[->, below] node{$\mathtt{0}$}  (q211)
    (q141) edge[->, above, near end] node{$\mathtt{1}$}  (q21-1)
       
    (q142) edge[->, above] node{$\mathtt{0}$}  (q211)
    (q142) edge[->, below,near start] node{$\mathtt{1}$}  (q21-1)
    
    (q21-1) edge[->, left] node{$\mathtt{0},\mathtt{1}$}  (q22-1)
       
    (q210) edge[->, left] node{$\mathtt{0},\mathtt{1}$}  (q220)
       
    (q211) edge[->, left,ultra thick] node{$\boldsymbol{0},\mathtt{1}$}  (q221)
       
    (q212) edge[->, left] node{$\mathtt{0},\mathtt{1}$}  (q222)
    
    (q22-1) edge[->, left] node{$\mathtt{0},\mathtt{1}$}  (q23-1)
       
    (q220) edge[->, above] node{$\mathtt{0}$}  (q23-1)
    (q220) edge[->, above] node{$\mathtt{1}$}  (q231)
       
    (q221) edge[->, above,ultra thick] node{$\boldsymbol{0}$}  (q23-1)
    (q221) edge[->] node[left=0.15cm]{$\mathtt{1}$}  (q232)
       
    (q222) edge[->, below,near start] node{$\mathtt{0}$}  (q23-1)
    (q222) edge[->, below] node{$\mathtt{1}$}  (q5)
    
    (q23-1) edge[->, left,ultra thick] node{$\boldsymbol{0},\mathtt{1}$}  (q24-1)
    
    (q230) edge[->, left] node{$\mathtt{0},\mathtt{1}$}  (q240)
       
    (q231) edge[->, left] node{$\mathtt{0},\mathtt{1}$}  (q241)
       
    (q232) edge[->, left] node{$\mathtt{0},\mathtt{1}$}  (q242)
    
    (q24-1) edge[->, loop below, below] node{$\mathtt{0},\mathtt{1}$}  (q24-1)
    
    (q240) edge[->, loop below, below] node{$\mathtt{0},\mathtt{1}$}  (q240)
       
    (q241) edge[->, loop below, below] node{$\mathtt{0},\mathtt{1}$}  (q241)
       
    (q242) edge[->, loop below, below] node{$\mathtt{0},\mathtt{1}$}  (q242)
    
    	;
    \node[] (q14-1label0) at ($(q14-1.south east)+(0.8,0.1)$) {$\boldsymbol{0}$};	
   	
\end{tikzpicture}
		\caption{Schematic drawing of the automaton~${A_\varphi}$ constructed
			for the 3SAT instance~$\varphi$, where~$\varphi=(x_1\wedge \bar{x}_2\wedge x_4)\vee (\overline{x}_1\wedge x_3)$. 
			The word $\lambda\in L({A_\varphi})$ is the only word in~$W(L({A_\varphi}))$ if and only if~$\varphi$ is
			a tautology.
			The unsatisfying assignment~$(0000)^3$ of~$\varphi$---which represents~$x_1=0$, $x_2=0$, $x_3=0$ and~$x_4=0$---is indicated with a boldface line.
		}
		\label{fig:conp-hardness-automaton-A}
	\end{figure}

\else
\begin{reptheorem}{thm:reg-reg-mpc-conp-hard}
	Let~$A$ and~$B$ be \dfa s. Then it is \co\NP-hard to decide
	whether $W^{\mpc}(L(A))=L(B)$ holds.
\end{reptheorem}
\begin{proof}
	In order to show the \co\NP-hardness of the \textsc{Witness-Language-Problem} we construct for given formula~$\varphi$ in 3DNF with~$n$ variables and~$m$ conjunctions a \dfa\ $A_\varphi$ which fulfils~$W(L(A_\varphi) \cap( \Sigma^{\neq 0 \bmod n} \cup \{\lambda\}))=\{\lambda\}$ if and only if~$\varphi$ is a tautology. (Notice that the intersection with $\Sigma^{\neq 0 \bmod n} \cup \{\lambda\}$ is necessary only for technical reasons to avoid deviant witnesses.)
	
	Assume that $\varphi=\bigvee_{i=1}^m C_i$ is a formula in 3DNF with~$n$
	variables $x_1,x_2,\ldots,x_n$ and~$m$ conjunctions. Define $\NonRep(n) = \Sigma^{\neq 0 \bmod n} \cup \{\lambda\}$ as the words over $\Sigma$ whose length is not dividable by $n$ but the empty word.
	
	Let ${A_\varphi}=(Q_{A_\varphi},\{\mathtt{0},\mathtt{1}\},{}\cdot_{A_\varphi}{},q_0,F_{A_\varphi})$ be the \dfa\, where
	$$	
	Q_{A_\varphi}=\{q_0,q_1,\ldots, q_{n+1}\}
	\cup \{q_{i,j,-1},q_{i,j,0},q_{i,j,1},q_{i,j,2}\mid 1\le i\le m, 1\le j\le n\}
	$$
	and~$F_{A_\varphi}=\{q_0,q_n,q_{n+1}\}\cup \{q_{i,n,k}\mid 1\le i< m,-1\le k\le 2\}$.
	
	We define its transitions 
	\begin{align}
		q_i\cdot_{A_\varphi} \mathtt{0} & = q_{i+1}, & & \mbox{for $0\leq i\leq n-1$},\label{eqn:pref-counter-start}\\
		q_i\cdot_{A_\varphi} \mathtt{1} & = q_{i+1}, & & \mbox{for $0\leq i\leq n-1$},\label{eqn:pref-counter-end}\\
		q_n\cdot_{A_\varphi} \mathtt{0} & = q_{1,1,1}, & & \mbox{if $\overline{x}_1\in C_1$},\label{eqn:first-if-i1=1-start}\\
		q_n\cdot_{A_\varphi} \mathtt{1} & = q_{1,1,-1}, & & \mbox{if $\overline{x}_1\in C_1$},\\
		q_n\cdot_{A_\varphi} \mathtt{0} & = q_{1,1,-1}, & & \mbox{if $x_1\in C_1$},\\
		q_n\cdot_{A_\varphi} \mathtt{1} & = q_{1,1,1}, & & \mbox{if $x_1\in C_1$}\label{eqn:first-if-i1=1-end},\\
		q_n\cdot_{A_\varphi} \mathtt{0} & = q_{1,1,0}, & & \mbox{if $\overline{x}_1\notin C_1\wedge x_1\notin C_1$},\label{eqn:x1-not-in-c1-0}\\
		q_n\cdot_{A_\varphi} \mathtt{1} & = q_{1,1,0},  & & \mbox{if $\overline{x}_1\notin C_1\wedge x_1\notin C_1$}.\label{eqn:x1-not-in-c1-1}\\
		\noalign{For $1\le i\le m$ and $1\le j\le n-1$ set}
		q_{i,j,k}\cdot_{A_\varphi} \mathtt{0} & = q_{i,j+1,k}, & & \mbox{if $\overline{x}_{j+1}\notin C_i\wedge x_{j+1}\notin C_i$ and~$0\le k\le 2$},\label{eqn:xj+1-not-in-ci-0}\\
		q_{i,j,k}\cdot_{A_\varphi} \mathtt{1} & = q_{i,j+1,k}, & & \mbox{if $\overline{x}_{j+1}\notin C_i\wedge x_{j+1}\notin C_i$ and~$0\le k\le 2$},\label{eqn:xj+1-not-in-ci-1}\\
		q_{i,j,k}\cdot_{A_\varphi} \mathtt{0} & = q_{i,j+1,k+1}, & & \mbox{if $\overline{x}_{j+1}\in C_i$ and~$0\le k\le |C_i|-1$,
		}\label{eqn:inside-count-up-negative}\\
		q_{i,j,k}\cdot_{A_\varphi} \mathtt{1} & = q_{i,j+1,-1}, & & \mbox{if $\overline{x}_{j+1}\in C_i$
			and~$0\le k\le |C_i|-1$,
		}\\
		q_{i,j,k}\cdot_{A_\varphi} \mathtt{0} & = q_{i,j+1,-1}, & & \mbox{if $x_{j+1}\in C_i$
			and~$0\le k\le |C_i|-1$,
		}\\
		q_{i,j,k}\cdot_{A_\varphi} \mathtt{1} & = q_{i,j+1,k+1}, & & \mbox{if $x_{j+1}\in C_i$ 
			and~$0\le k\le |C_i|-1$,
		}\label{eqn:inside-count-up-positiv}\\
		q_{i,j,k}\cdot_{A_\varphi} \mathtt{0} & = q_{n+1}, & & \mbox{if $\overline{x}_{j+1}\in C_i$ and~$|C_i|\le k\le 2$,
		}\label{eqn:inside-count-up-negative-greatest}\\
		q_{i,j,k}\cdot_{A_\varphi} \mathtt{1} & = q_{i,j+1,-1}, & & \mbox{if $\overline{x}_{j+1}\in C_i$ and~$|C_i|\le k\le 2$,
		}\\
		q_{i,j,k}\cdot_{A_\varphi} \mathtt{0} & = q_{i,j+1,-1}, & & \mbox{if $x_{j+1}\in C_i$ and~$|C_i|\le k\le 2$,
		}\\
		q_{i,j,k}\cdot_{A_\varphi} \mathtt{1} & = q_{n+1}, & & \mbox{if $x_{j+1}\in C_i$ and~$|C_i|\le k\le 2$,
		}\label{eqn:inside-count-up-positiv-greatest}\\
		q_{i,j,2}\cdot_{A_\varphi} \mathtt{0} & = q_{n+1}, & & \mbox{if $\overline{x}_{j+1}\in C_i$},\label{eqn:final-0}\\
		q_{i,j,2}\cdot_{A_\varphi} \mathtt{1} & = q_{n+1}, & & \mbox{if $x_{j+1}\in C_i$},\label{eqn:final-1}\\
		q_{i,j,-1}\cdot_{A_\varphi} \mathtt{0} & = q_{i,j+1,-1}, & & \label{eqn:minus-one-0}\\
		q_{i,j,-1}\cdot_{A_\varphi} \mathtt{1} & = q_{i,j+1,-1}. & & \label{eqn:minus-one-1}\\
		\noalign{For $1\le i\le m-1$ and $-1\le k\le 2$ set}
		q_{i,n,k}\cdot_{A_\varphi} \mathtt{0} & = q_{i+1,1,0}, & & \mbox{if $\overline{x}_1\notin C_{i+1}\wedge x_1\notin C_{i+1}$},\label{eqn:x1-not-in-ci+1-0}\\
		q_{i,n,k}\cdot_{A_\varphi} \mathtt{1} & = q_{i+1,1,0},  & & \mbox{if $\overline{x}_1\notin C_{i+1}\wedge x_1\notin C_{i+1}$},\label{eqn:x1-not-in-ci+1-1}\\
		q_{i,n,k}\cdot_{A_\varphi} \mathtt{0} & = q_{i+1,1,1}, & & \mbox{if $\overline{x}_1\in C_{i+1}$},\label{eqn:counter-too-small-start}\\
		q_{i,n,k}\cdot_{A_\varphi} \mathtt{1} & = q_{i+1,1,-1}, & & \mbox{if $\overline{x}_1\in C_{i+1}$},\\
		q_{i,n,k}\cdot_{A_\varphi} \mathtt{0} & = q_{i+1,1,-1}, & & \mbox{if $x_1\in C_{i+1}$},\\
		q_{i,n,k}\cdot_{A_\varphi} \mathtt{1} & = q_{i+1,1,1}, & & \mbox{if $x_1\in C_{i+1}$}\label{eqn:counter-too-small-end}.\\
		\noalign{For $-1\le k\le 2$ set}
		q_{m,n,k}\cdot_{A_\varphi} \mathtt{0} & = q_{m,n,k}, & &\\
		q_{m,n,k}\cdot_{A_\varphi} \mathtt{1} & = q_{m,n,k}, & &\\
		q_{n+1}\cdot_{A_\varphi} \mathtt{0} & = q_{n+1}, & &\\
		q_{n+1}\cdot_{A_\varphi} \mathtt{1} & = q_{n+1}. & &
	\end{align}	
	The automaton~${A_\varphi}$ is shown schematically in Figure~\ref{fig:conp-hardness-automaton-A}, where~$\varphi=(x_1\wedge \bar{x}_2\wedge x_4)\vee (\overline{x}_1\wedge x_3)$.
	
		\begin{figure}[!htb]
		\centering
		\begin{tikzpicture}[->,>=stealth',shorten >=1pt,auto,node distance=2.9cm,
	semithick,scale=0.7, every node/.style={scale=0.7},initial text=]
	\tikzstyle{every state}=[fill={rgb:black,1;white,10},ellipse]

	\node[state,initial,accepting] (q0) at (7,20) {$q_0$};	
	\node[state] (q1) at (7,18.5)  {$q_1$};
	\node[state] (q2) at (7,17)  {$q_2$};	
	\node[state] (q3) at (7,15.5)  {$q_3$};	
	\node[state,accepting] (q4) at (7,14)  {$q_4$};	
	
	\node[state] (q11-1) at (2.5,12)  {$q_{1,1,-1}$};	
	\node[state] (q110) at (5.5,12)  {$q_{1,1,0}$};	
	\node[state] (q111) at (8.5,12)  {$q_{1,1,1}$};	
	\node[state] (q112) at (11.5,12)  {$q_{1,1,2}$};	
	
    \node[state] (q12-1) at (2.5,10)  {$q_{1,2,-1}$};	
	\node[state] (q120) at (5.5,10)  {$q_{1,2,0}$};	
	\node[state] (q121) at (8.5,10)  {$q_{1,2,1}$};	
	\node[state] (q122) at (11.5,10)  {$q_{1,2,2}$};	
	
	\node[state] (q13-1) at (2.5,8)  {$q_{1,3,-1}$};	
	\node[state] (q130) at (5.5,8)  {$q_{1,3,0}$};	
	\node[state] (q131) at (8.5,8)  {$q_{1,3,1}$};	
	\node[state] (q132) at (11.5,8)  {$q_{1,3,2}$};	
	
    \node[state,accepting] (q14-1) at (2.5,6)  {$q_{1,4,-1}$};	
	\node[state,accepting] (q140) at (5.5,6)  {$q_{1,4,0}$};	
	\node[state,accepting] (q141) at (8.5,6)  {$q_{1,4,1}$};	
	\node[state,accepting] (q142) at (11.5,6)  {$q_{1,4,2}$};	

    \node[state] (q21-1) at (2.5,4)  {$q_{2,1,-1}$};	
	\node[state] (q210) at (5.5,4)  {$q_{2,1,0}$};	
	\node[state] (q211) at (8.5,4)  {$q_{2,1,1}$};	
	\node[state] (q212) at (11.5,4)  {$q_{2,1,2}$};	
	
	\node[state] (q22-1) at (2.5,2)  {$q_{2,2,-1}$};	
	\node[state] (q220) at (5.5,2)  {$q_{2,2,0}$};	
	\node[state] (q221) at (8.5,2)  {$q_{2,2,1}$};	
	\node[state] (q222) at (11.5,2)  {$q_{2,2,2}$};	
	
    \node[state] (q23-1) at (2.5,0)  {$q_{2,3,-1}$};	
	\node[state] (q230) at (5.5,0)  {$q_{2,3,0}$};	
	\node[state] (q231) at (8.5,0)  {$q_{2,3,1}$};	
	\node[state] (q232) at (11.5,0)  {$q_{2,3,2}$};	
	
    \node[state] (q24-1) at (2.5,-2)  {$q_{2,4,-1}$};	
	\node[state] (q240) at (5.5,-2)  {$q_{2,4,0}$};	
	\node[state] (q241) at (8.5,-2)  {$q_{2,4,1}$};	
	\node[state] (q242) at (11.5,-2)  {$q_{2,4,2}$};	
	
	\node[state,accepting] (q5) at (14,5){$q_5$};
	
	\node[] (qa) at (-1,13) {$x_1\in C_1$};
	
	\node[] (qb) at (-1,11) {$\overline{x}_2\in C_1$};
	
	\node[] (qc) at (-1,9) {$x_3,\overline{x}_3\not\in C_1$};
	
	\node[] (qd) at (-1,7) {$x_4\in C_1$};
	
	\node[] (qe) at (-1,5) {$\overline{x}_1\in C_2$};
	
	\node[] (qf) at (-1,3) {$x_2,\overline{x}_2\not\in C_2$};
	
	\node[] (qg) at (-1,1) {$x_3\in C_2$};
	
	\node[] (qh) at (-1,-1) {$x_4,\overline{x}_4\not\in C_2$};
	
\draw   
    (q0) edge[ultra thick, ->,left] node{$\boldsymbol{0},\mathtt{1}$}       (q1)
    
    (q1) edge[->,left,ultra thick] node{$\boldsymbol{0},\mathtt{1}$}       (q2)
            
    (q2) edge[->, left,ultra thick] node{$\boldsymbol{0},\mathtt{1}$}     (q3)
    
    (q3) edge[->, left,ultra thick] node{$\boldsymbol{0},\mathtt{1}$}  (q4)
            
    (q4) edge[->, above,ultra thick] node{$\boldsymbol{0}$}  (q11-1)
    (q4) edge[->, above] node{$\mathtt{1}$}  (q111)
    
    (q11-1) edge[->, left,ultra thick] node{$\boldsymbol{0},\mathtt{1}$}  (q12-1)
    
    (q110) edge[->, above] node{$\mathtt{0}$}  (q121)
    (q110) edge[->, above] node{$\mathtt{1}$}  (q12-1)
    
    (q111) edge[->, left,above] node{$\mathtt{0}$}  (q122)
    (q111) edge[->, above] node{$\mathtt{1}$}  (q12-1)
    
    (q112) edge[->, bend left, right] node{$\mathtt{0}$}  (q5)
    (q112) edge[->, near start,below] node{$\mathtt{1}$}  (q12-1)
    
    (q12-1) edge[->, left,ultra thick] node{$\boldsymbol{0},\mathtt{1}$}  (q13-1)
    
    (q120) edge[->, left] node{$\mathtt{0},\mathtt{1}$}  (q130)
    
    (q121) edge[->, left] node{$\mathtt{0},\mathtt{1}$}  (q131)
    
    (q122) edge[->, left] node{$\mathtt{0},\mathtt{1}$}  (q132)
    
    (q13-1) edge[->, left,ultra thick] node{$\boldsymbol{0},\mathtt{1}$}  (q14-1)
    
    (q130) edge[->, above] node{$\mathtt{0}$}  (q14-1)
    (q130) edge[->, above] node{$\mathtt{1}$}  (q141)
    
    (q131) edge[->, above] node{$\mathtt{0}$}  (q14-1)
    (q131) edge[->, left] node{$\mathtt{1}$}  (q142)
    
    (q132) edge[->, below, near start] node{$\mathtt{0}$}  (q14-1)
    (q132) edge[->, left] node{$\mathtt{1}$}  (q5)
    
    (q14-1) edge[->,above, ultra thick] node[name=q14-10]{}  (q211)

    (q14-1) edge[->, left] node{$\mathtt{1}$}  (q21-1)
       
    (q140) edge[->, above,near end] node{$\mathtt{0}$}  (q211)
    (q140) edge[->, above, near end] node{$\mathtt{1}$} (q21-1)
       
    (q141) edge[->, below] node{$\mathtt{0}$}  (q211)
    (q141) edge[->, above, near end] node{$\mathtt{1}$}  (q21-1)
       
    (q142) edge[->, above] node{$\mathtt{0}$}  (q211)
    (q142) edge[->, below,near start] node{$\mathtt{1}$}  (q21-1)
    
    (q21-1) edge[->, left] node{$\mathtt{0},\mathtt{1}$}  (q22-1)
       
    (q210) edge[->, left] node{$\mathtt{0},\mathtt{1}$}  (q220)
       
    (q211) edge[->, left,ultra thick] node{$\boldsymbol{0},\mathtt{1}$}  (q221)
       
    (q212) edge[->, left] node{$\mathtt{0},\mathtt{1}$}  (q222)
    
    (q22-1) edge[->, left] node{$\mathtt{0},\mathtt{1}$}  (q23-1)
       
    (q220) edge[->, above] node{$\mathtt{0}$}  (q23-1)
    (q220) edge[->, above] node{$\mathtt{1}$}  (q231)
       
    (q221) edge[->, above,ultra thick] node{$\boldsymbol{0}$}  (q23-1)
    (q221) edge[->] node[left=0.15cm]{$\mathtt{1}$}  (q232)
       
    (q222) edge[->, below,near start] node{$\mathtt{0}$}  (q23-1)
    (q222) edge[->, below] node{$\mathtt{1}$}  (q5)
    
    (q23-1) edge[->, left,ultra thick] node{$\boldsymbol{0},\mathtt{1}$}  (q24-1)
    
    (q230) edge[->, left] node{$\mathtt{0},\mathtt{1}$}  (q240)
       
    (q231) edge[->, left] node{$\mathtt{0},\mathtt{1}$}  (q241)
       
    (q232) edge[->, left] node{$\mathtt{0},\mathtt{1}$}  (q242)
    
    (q24-1) edge[->, loop below, below] node{$\mathtt{0},\mathtt{1}$}  (q24-1)
    
    (q240) edge[->, loop below, below] node{$\mathtt{0},\mathtt{1}$}  (q240)
       
    (q241) edge[->, loop below, below] node{$\mathtt{0},\mathtt{1}$}  (q241)
       
    (q242) edge[->, loop below, below] node{$\mathtt{0},\mathtt{1}$}  (q242)
    
    	;
    \node[] (q14-1label0) at ($(q14-1.south east)+(0.8,0.1)$) {$\boldsymbol{0}$};	
   	
\end{tikzpicture}
		\caption{Schematic drawing of the automaton~${A_\varphi}$ constructed
			for the 3SAT instance~$\varphi$, where~$\varphi=(x_1\wedge \bar{x}_2\wedge x_4)\vee (\overline{x}_1\wedge x_3)$. 
			The word $\lambda\in L({A_\varphi})$ is the only word in~$W(L({A_\varphi})\cap( \Sigma^{\neq 0 \bmod n} \cup \{\lambda\}))$ if and only if~$\varphi$ is
			a tautology.
			The unsatisfying assignment~$(0000)^3$ of~$\varphi$---which represents~$x_1=0$, $x_2=0$, $x_3=0$ and~$x_4=0$---is indicated with a boldface line.
		}
		\label{fig:conp-hardness-automaton-A}
	\end{figure}
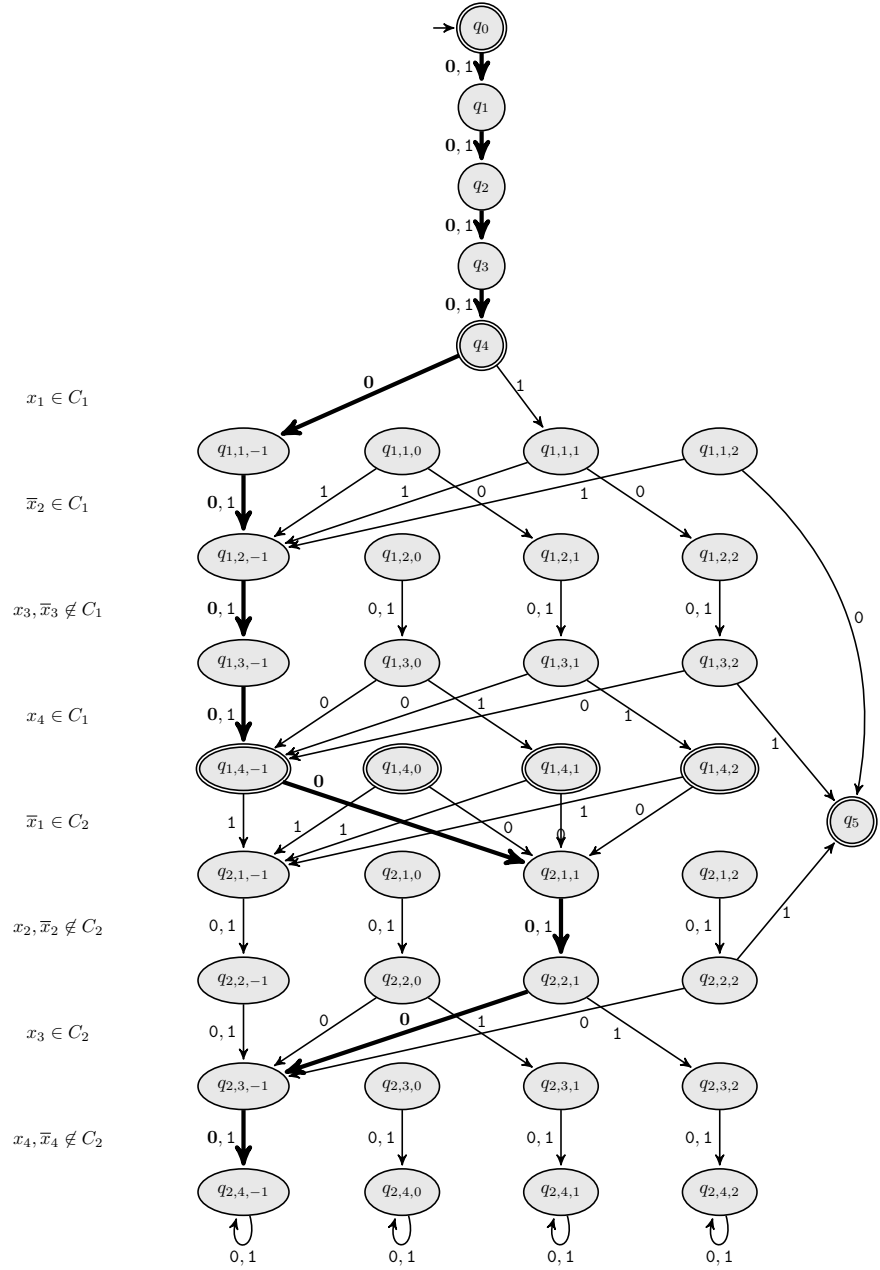

	Each assignment for a conjunction~$C_i$ is represented by a word~$w\in \{\mathtt{0},\mathtt{1}\}^n$ in the obvious way, i.e., $\mathtt{0}$ represents the value \textit{false} and~$\mathtt{1}$ the value \textit{true}.
	Indeed the clause~$C_i$ is evaluated to true if each of its variables is assumed to be \textit{true}. 
	
	\begin{mialc}\label{clm:clause-Assignments-to-accepting-sink}
		Given a conjunction~$C_i$ of~$\varphi$ and an assignment of~$\varphi$ represented by a word~$w\in\{\mathtt{0},\mathtt{1}\}^n$.
		The states~$q_{i-1,1,k}$ for~$-1\le k\le 2$ and~$i>1$ and~$q_n$ if~$i=1$
		are mapped onto the accepting sink state~$q_{n+1}$ by~$w$
		if and only if~$C_i$ evaluates to \textit{true} for the assignment~$w$.
	\end{mialc}

	\begin{proof}
		Since~$\varphi$ is in 3DNF there are at most three variables~$x_{i_1},x_{i_2},x_{i_3}$ for each conjunction~$C_i$ such that~$C_i=y_{i_1}\wedge y_{i_2}\wedge y_{i_3}$ for~$y_{i_1}\in\{x_{i_1},\overline{x}_{i_1}\},y_{i_2}\in\{x_{i_2},\overline{x}_{i_2}\},$ and~$y_{i_3}\in\{x_{i_3},\overline{x}_{i_3}\}$
		for indices~$i_1,i_2,$ and~$i_3$ with~$1\le i_1<i_2<i_3\le n$.
		
		First we study the case where~$C_i$ contains exactly three literals.
		Observe that for all indices~$j<i_1$ we have~$\{x_j,\overline{x}_j\}\cap C_i= \emptyset$. Hence, the rules~\ref{eqn:x1-not-in-c1-0}, \ref{eqn:x1-not-in-c1-1} for~$q_n$,
		\ref{eqn:x1-not-in-ci+1-0}, \ref{eqn:x1-not-in-ci+1-1} for~$q_{i-1,1,k}$,
		and~\ref{eqn:xj+1-not-in-ci-0}, \ref{eqn:xj+1-not-in-ci-1} imply that 
		$q_n\cdot_{A_\varphi} w_1w_2\ldots w_{i_1-1}=q_{1,i_1-1,0}$ and
		$q_{i-1,1,0}\cdot_{A_\varphi} w_1w_2\ldots w_{i_1-1}=q_{i,i_1-1,0}$.
		
		Next the state~$q_{i,i_1-1,0}$ is mapped by~$w_{i_1}\in\{\mathtt{0},\mathtt{1}\}$ onto
		$$\begin{cases}
			q_{i,i_1,-1}, & \mbox{if~$y_{i_1}=x_{i_1}$ and~$w_{i_1}=\mathtt{0}$,}\\
			q_{i,i_1,1}, & \mbox{if~$y_{i_1}=x_{i_1}$ and~$w_{i_1}=\mathtt{1}$,}\\
			q_{i,i_1,1}, & \mbox{if~$y_{i_1}=\overline{x}_{i_1}$ and~$w_{i_1}=\mathtt{0}$,}\\
			q_{i,i_1,-1}, & \mbox{if~$y_{i_1}=\overline{x}_{i_1}$ and~$w_{i_1}=\mathtt{1}$},
		\end{cases}$$
		due to the transition rules
		\ref{eqn:inside-count-up-negative}-\ref{eqn:inside-count-up-positiv}.
		Hence, the counter~$k$ is increased by one if the variable~$x_{i_1}$ is assumed to be \textit{true} and it appears as a positive literal in~$C_i$ or the variable~$x_{i_1}$ is assumed to be \textit{false} and it appears as a negative literal in~$C_i$.
		Again by applying the transition rules~\ref{eqn:xj+1-not-in-ci-0} and \ref{eqn:xj+1-not-in-ci-1} we observe that~$q_{i,i_1,1}$ is mapped by~$w_{i_1+1}w_{i_1+2}\ldots w_{i_2-1}$ onto~$q_{i,i_2-1,1}$.
		Analogously as for~$i_1$ the state~$q_{i,i_2-1,1}$ is mapped by~$w_{i_2}$ onto 
		$$\begin{cases}
			q_{i,i_2,-1}, & \mbox{if~$y_{i_2}=x_{i_2}$ and~$w_{i_2}=\mathtt{0}$,}\\
			q_{i,i_2,2}, & \mbox{if~$y_{i_2}=x_{i_2}$ and~$w_{i_2}=\mathtt{1}$,}\\
			q_{i,i_2,2}, & \mbox{if~$y_{i_2}=\overline{x}_{i_2}$ and~$w_{i_2}=\mathtt{0}$,}\\
			q_{i,i_2,-1}, & \mbox{if~$y_{i_2}=\overline{x}_{i_2}$ and~$w_{i_2}=\mathtt{1}$},
		\end{cases}$$
		due to the transition rules
		\ref{eqn:inside-count-up-negative}-\ref{eqn:inside-count-up-positiv} if~$i_1>1$ or~$i>1$, and due to~\ref{eqn:first-if-i1=1-start}-\ref{eqn:first-if-i1=1-end} otherwise.
		Proceeding with this procedure the state~$q_{i,i_2,1}$ is mapped by~$w_{i_2+1}w_{i_2+2}\ldots w_{i_3-1}$ onto~$q_{i,i_3-1,1}$ and then due to
		the rules~\ref{eqn:final-0} and~\ref{eqn:final-1} it is mapped by~$w_{i_3}$ onto the state~$q_{n+1}$. Since~$q_{n+1}$ is an accepting sink state the word~$w_{i_3+1}w_{i_3+2}\ldots w_{|w|}$ maps it onto itself and the claim follows.
		
		If~$C_i$ contains less than three literals the statement follows in similar vein as in the previous case but the state~$q_{i,i_2-1,1}$ (for two literals) and the state~$q_{i,i_1-1,0}$ (for one literal) is mapped directly onto
		$q_{n+1}$ if~$y_{i_2}$ and~$y_{i_1}$, \resp, is assumed to be true
		due to the transition rules
		\ref{eqn:inside-count-up-negative-greatest}-\ref{eqn:inside-count-up-positiv-greatest}.
		\iflipics\else\qed\fi
	\end{proof}
	
	Since an assignment for $\varphi$ maps each variable of~$\varphi$ to a unique truth value, each variable has the same truth value in each conjunction of~$\varphi$.
	This implies that each assignment for~$\varphi$ can be represented\footnote{Indeed, it would be sufficient to use~$w^m$ but for controlling the set of witness words we need~$w^{m+1}$ which is also capable of representing all assignments of~$\varphi$.} by a word~$w^{m+1}$, where~$w\in\{\mathtt{0},\mathtt{1}\}^n$.
	
	\begin{mialc}\label{clm:varphi-Assignments-to-accepting-sink}
		Given an assignment for~$\varphi$ represented by~$w^{m+1}$ for~$w\in\{\mathtt{0},\mathtt{1}\}^n$.
		The state~$q_0$ is mapped onto the accepting sink state~$q_{n+1}$ by~$w^{m+1}$
		if~$\varphi$ evaluates to \textit{true} for the assignment~$w^{m+1}$
		and onto a non-accepting sink state~$q_{m,n,k}$ for~$-2\le k\le 2$ otherwise.
	\end{mialc}
	
	\begin{proof}
		Due to the rules~\ref{eqn:pref-counter-start} and~\ref{eqn:pref-counter-end} we obtain that~$q_0$ is mapped by~$w$ onto~$q_n$.
		Then we can apply Claim~\ref{clm:clause-Assignments-to-accepting-sink} such that~$q_n$ is eventually mapped onto~$q_{n+1}$ or onto~$q_{1,j,-1}$.
		Additionally, each state~$q_{i,j,-1}$ is mapped onto~$q_{i,j+1,-1}$ by~$\mathtt{0}$ and~$\mathtt{1}$ (due to~\Cref{eqn:minus-one-0} and~\ref{eqn:minus-one-1}) and thus finally onto~$q_{i,n,-1}$ by~$w_{j+1},w_{j+2},\ldots, w_n$.
		Hence, for each assignment~$w^m$ which does not evaluate~$C_i$ to \textit{true} we have for~$-1\le k\le 2$ that~$q_{i-1,n,k}$ for~$i>1$ ($q_n$ \resp, for~$i=1$) is mapped by~$w$ onto~$q_{i,n,-1}$.
		By repeating this argument together with Claim~\ref{clm:clause-Assignments-to-accepting-sink} the stated claim follows.
		\iflipics\else\qed\fi
	\end{proof}
	
		Next we consider the witness language induced by the language accepted by~${A_\varphi}$.
	
	\begin{mialc}\label{clm:tautology-equals-only-one-witness}
		Let~$L$ be the language accepted by the \dfa~${A_\varphi}$. Then
		$W^{\mpc}(L({A_\varphi}) \cap \NonRep(n))=\{\lambda\}$ if and only if the Boolean formula~$\varphi$ is a tautology.
	\end{mialc}
	
	\begin{proof}
		First we observe that Claim~\ref{clm:varphi-Assignments-to-accepting-sink} implies that the words~$w$ of length~$n$ are in~$W(L({A_\varphi}) \cap \NonRep(n))$, i.e., $w$ is not pumpable, if and only if~$w^m$ is an unsatisfying assignment of~$\varphi$.
		
		Additionally, for an satisfying assignment~$w^m$ we have that~$w^m v$ is pumpable by~$w$ for each word~$v\in\{\mathtt{0},\mathtt{1}\}^*$ because there is a conjunction~$C_i$ such that~$w$ is a satisfying assignment for~$C_i$ which implies that~$w^{i+1}$ maps~$q_0$ onto the accepting sink state~$q_{n+1}$.
		If~$\varphi$ is a tautology we obtain that~$L({A_\varphi})\subseteq\{\mathtt{0},\mathtt{1}\}^n\{\mathtt{0},\mathtt{1}\}^*\cup\{\lambda\}$ and each word in~$\{\mathtt{0},\mathtt{1}\}^n\{\mathtt{0},\mathtt{1}\}^*$ is equal to~$wv$ for~$w\in\{\mathtt{0},\mathtt{1}\}^n$ an~$v\in\{\mathtt{0},\mathtt{1}\}^*$ such that~$w^m$ is a satisfying assignment of~$\varphi$.
		Hence, all words in~$(L({A_\varphi})\cap \NonRep(n))\setminus \{\lambda\}$ are pumpable by their suffix of length~$\operatorname{mod} n$ which implies that~$W^{\mpc}(L({A_\varphi})\cap \NonRep(n))=\{\lambda\}$.
		
		On the other hand, let~$w^m$ be an unsatisfying assignment of~$\varphi$. Then~$q_0$
		is mapped by~$w^m$ onto one of the non-accepting sink states~$q_{m,n,k}$ for~$-1\le k\le 2$ due to Claim~\ref{clm:varphi-Assignments-to-accepting-sink}.
		Since~$w\in L({A_\varphi})\cap \NonRep(n)$ and $\lambda$ is the only factor of~$w$ which is in~$L({A_\varphi})\cap \NonRep(n)$, we obtain that~$w$ cannot be pumped by any of its factors.
		Therefore, we have~$w\in W(L({A_\varphi}) \cap \NonRep(n))$ which implies that~$W(L({A_\varphi})\cap \NonRep(n))\ne \{\lambda\}$.
		\iflipics\else\qed\fi
	\end{proof}

	Let~$B$ be the \dfa\ accepting the language~$\{\lambda\}$.
	Therefore, the formula~$\varphi$ is a tautology if and only if~$\langle {A_\varphi},B\rangle$ is a yes instance of the \textsc{Witness-Language-Problem}
	w.r.t. Lemma~\ref{lem:mpc} (Lemma~\ref{lem:mpl}, \resp).
	We observe that the \dfa s~${A_\varphi}$ and~$B$  as well as the \dfa  ~for the language $\NonRep(n)$ can be constructed in polytime from~$\varphi$ and they have~$(m+1)\cdot n+2$, $2$ and $n+1$ states which are polynomials in the size of~$\varphi$.
	Therefore, the previously shown construction\footnote{Since we have~$\mpc(L({A_\varphi}))=1$ the construction re-establishes the \co\NP-hardness of the \textsc{Pumping-Problem}.} is a polytime reduction from \textsc{TAUTOLOGY} onto the \textsc{Witness-Language-Problem}.
	Thus, the \textsc{Witness-Language-Problem} for \dfa s w.r.t. Lemma~\ref{lem:mpc} is \co\NP-hard.
	\iflipics\else\qed\fi
	\end{proof}
\fi

Putting the results together, we get:

\begin{corollary}\label{cor:witness-language-problem-conp-complete}
	Let~$A$ and~$B$ be \dfa s. Then it is \co\NP-complete to decide
	whether~$W(L(A))=L(B)$ holds w.r.t.\ Lemma~\ref{lem:mpc}.\qed
\end{corollary}

Note that Theorem~\ref{thm:reg-reg-mpc-conp-hard} does not directly apply to $\mpl$ 
 as the reduction relies on the fact that words are pumpable by their suffix of length~$n+(|w|\bmod n)$. This is only possible for~$\mpl$ if~$\mpl(L(A_{\varphi}))\ge|w|$. Since all assignments which fulfil only the last conjunction cannot be pumped by their prefix of length~$nm$ we have~$\mpl(L(A_{\varphi}))\ge nm$ and thus~$W^{\mpl}(L(A_{\varphi}))$ does not contain any words shorter than~$nm$ independently of the choice of~$\varphi$.



	
	\section{Conclusions}
	In this work, we continued the investigation of witness languages first introduced in \cite{Ra25}. 
First, we examined the witness property for powers of words within pumping lemmata for non-regular languages and proved that for every non-regular language with infinite scattered factor distance, there exists a word that witnesses the non-regularity of this language by pumping with $t = 0$ w.r.t.~\Cref{lem:mpl}. 
As a second step, we pursued with research on regular languages. 
We established several language-theoretic properties of witness languages, including the result that every witness language of a regular language is itself regular. 
Furthermore, we derived bounds on the state complexity of witness languages. As the concept of witnesses is rather intricate in the general case, we subsequently restricted our attention to unary witness languages, for which we provided precise descriptions of witness languages alongside combinatorial properties on witnesses. 
In addition, we introduced and investigated the \textsc{Pumping-Witness-Language-Problem}, namely the decision problem whether a given regular language is the witness language of another language, and we proved that this problem is \coNP-complete. From this, we conclude that the decision problems $W(L) \subseteq W(L')$, $W(L) \cap W(L') = \emptyset$, and $W(L) = W(L')$ for regular languages $L,L'$ are all \coNP-complete for \Cref{lem:mpc} (w.r.t.~$\mpc$). 
	
	\newpage

	\bibliography{shortcortana,cortana}

@Preamble{"\def\lastmodified{23.12.2008}"}

@Preamble{"\def\id#1{#1}"}

@article{B,
	author = 	{R. Greenlaw and H. J. Hoover and W. L. Ruzzo},
	title = 	{Limits to Parallel Computation; {P}-Completeness Theory}
}

@book{Br84,
	author = 	{W. Brauer},
	title = 	{Automatentheorie: {E}ine {E}inf{\"u}hrung in die {T}heorie endlicher {A}utomaten},
	publisher = 	{Teubner Stuttgart},
	year = 		1984,
	series = 	{Leitf{\"a}den und {M}onographien der {I}nformatik},
	note =		{(in German)}
}

@book{Ha78,
   author = "M. A. Harrison",
   publisher = "Addison-Wesley",
   title = "Introduction to Formal Language Theory",
   year = "1978"
}

@book{HoUl79,
   	author = "J. E. Hopcroft and J. D. Ullman",
   	publisher = "Addison-Wesley",
   	title = "Introduction to Automata Theory, Languages and Computation",
   	year = "1979"
}

@article{Ni82,
	author = 	{A. Nijholt},
	title = 	{{YABBER}---Yet Another Bibliography: Pumping Lemma's. An Annotated Bibliography of Pumping},
	journal = 	{Bulletin of the European Association for Theoretical Computer Science},
	volume = 	{\textbf{17}},
	pages = 	{35--53},
	year = 		1982
}

@article{RaSc59,
	author = 	{M. O. Rabin and D. Scott},
	title = 	{Finite Automata and Their Decision Problems},
	journal = 	{IBM},
	volume = 	{\textbf{3}},
	pages = 	{114--125},
	year = 		1959
}

@article{Hi52,
	author = 	{G. Higman},
	title = 	{Ordering by Divisibility in Abstract Algebras},
	journal = 	{Proceedings of the London Mathematical Society},
	volume = 	3,
	number = 	2,
	pages = 	{326--336},
	year = 		1952
}

@book{HaWr79,
	author = 	{G. H. Hardy and E. M. Wright},
	title = 	{An Introduction to the Theory of Numbers},
	publisher = 	{Oxford Science Publications},
	year = 		{1979}
}

@inproceedings{GrHoRa24b,
  author       = {H. Gruber and M. Holzer and C. Rauch},
  editor       = {J. D. Day and
                  F. Manea},
  title        = {The Pumping Lemma for Context-Free Languages is Undecidable},
  booktitle    = {{DLT}},
  series       = {LNCS},
  volume       = {\textbf{14791}},
  pages        = {141--155},
  publisher    = {Springer},
  year         = {2024}
}

@inproceedings{GrHoRa24c,
  author       = {H. Gruber and M. Holzer and C. Rauch},
  title        = {On Pumping Problems for Unary Regular Languages},
  editor       = {R. Kr{\'{a}}lovic and
                  V. Kurkov{\'{a}}},
  
  booktitle    = {{SOFSEM}},
  series       = {LNCS},
  volume       = {15539},
  pages        = {3--16},
  publisher    = {Springer},
  year         = {2025}
}

@article{Ja78,
	author = 	{J. Jaffe},
	title = 	{A necessary and sufficient pumping lemma for regular languages},
	journal = 	{SIGACT News},
	volume = 	{\textbf{10}},
	number = 	{2},
	pages = 	{48--49},
	year = 		1978
}

@book{Ko97,
	author = 	{D. C. Kozen},
	title = 	{Automata and Computability},
	publisher = 	{Springer},
	year = 		1997,
	series = 	{Undergraduate Texts in Computer Science}
}

@article{DaJe22,
	author = 	{J. Dassow and I. Jecker},
	title = 	{Operational complexity and pumping lemmas},
	journal = 	{AI},
	volume = 	{\textbf{59}},
	pages = 	{337-355},
	year = 		{2022}
}

@inproceedings{HoRa23a,
        author =        {M. Holzer and C. Rauch},
	title =         {On {J}affe’s Pumping Lemma, Revisited},
        booktitle =     {Proceedings of the $25$th International Conference on Descriptional Complexity of Formal Systems},
        month =         jul,
        year =          2023,
	pages =		{65--78},
        editor =        {H. Bordihn and N. Tran and G. Vaszil},
        address =       {Potsdam, Germany},
        series =        {LNCS},
	number =	13918,
        publisher =     {Springer}
}

@inproceedings{GHR23,
        author =        {H. Gruber and M. Holzer and C. Rauch},
	title =         {The Pumping Lemma for Regular Languages is Hard},
        booktitle =     {Proceedings of the $27$th International Conference on Implementation and Application of Automata},
        month =         sep,
        year =          2023,
	pages =		{128--140},
        address =       {Famagusta, Cyprus},
        series =        {LNCS},
	number =	14151,
        publisher =     {Springer}
}

@inproceedings{HoRa23d,
        author =        {M. Holzer and C. Rauch},
        title =         {On Minimal Pumping Constants for Regular Languages},
        booktitle =     {Proceedings of the $16$th International Conference on Automata and Formal Languages},
        month =         sep,
        year =          2023,
        pages =         {127--141},
        number =        386,
        editor =        {Z. Gazdag and S. Iv{\'a}n and G. Kov{\'a}sznai},
        address =       {Eger, Hungary},
        series =        {EPTCS}
}

@inproceedings{Ka14,
	author = 	{D. Kaloci{\'n}ski},
	title = 	{On Computability and Learnability of the Pumping Lemma Function},
	booktitle = 	{Proceedings of the $8$th International Conference Language and Automata Theory and Applications},
	year = 		2014,
	address = 	{Madrid, Spain},
	pages = 	{433--440},
	series = 	{LNCS},
	number = 	8370,
	publisher = 	{Springer}
}

@article{Ho78,
	author = 	{S. Horv{\'a}th},
	title = 	{The family of languages satisfying {B}ar-{H}illel’s lemma},
	journal = 	{RAIRO},
	volume = 	{\textbf{12}},
	number = 	{3},
	pages = 	{193--199},
	year = 		1978
}

@article{So80,
	author = 	{R. Sommerhalder},
	title = 	{Classes of languages proof against regular pumping},
	journal = 	{RAIRO},
	volume = 	{\textbf{14}},
	number = 	{2},
	pages = 	{169--18},
	year = 		{1980}
}

@inproceedings{MPPS20,
  author       = {F. Mr{\'{a}}z and
                  D. Pardubsk{\'{a}} and
                  M. Pl{\'{a}}tek and
                  J. {\v S}{\'{i}}ma},
  editor       = {M. Holena and
                  T.s Horv{\'{a}}th and
                  A. Kelemenov{\'{a}} and
                  F. Mr{\'{a}}z and
                  D. Pardubsk{\'{a}} and
                  M. Pl{\'{a}}tek and
                  P. Sos{\'{\i}}k},
  title        = {Pumping Deterministic Monotone Restarting Automata and {DCFL}},
  booktitle    = {Proceedings of the 20th Conference Information Technologies - Applications
                  and Theory {(ITAT} 2020), Hotel Tyrapol, Oravsk{\'{a}} Lesn{\'{a}},
                  Slovakia, September 18-22, 2020},
  series       = {{CEUR} Workshop Proceedings},
  volume       = {\textbf{2718}},
  pages        = {51--58},
  publisher    = {CEUR-WS.org},
  year         = {2020}
}

@inproceedings{PMPPS21,
  author       = {M. Pl{\'{a}}tek and
                  F. Mr{\'{a}}z and
                  D. Pardubsk{\'{a}} and
                  D. Pr{\r u}{\v s}a and
                  J. {\v S}{\'{i}}ma},
  editor       = {B. Brejov{\'{a}} and
                  L. Ciencialov{\'{a}} and
                  M. Holena and
                  F. Mr{\'{a}}z and
                  D. Pardubsk{\'{a}} and
                  M. Pl{\'{a}}tek and
                  T. Vinar},
  title        = {On Separations of LR(0)-Grammars by Two Types of Pumping Patterns},
  booktitle    = {Proceedings of the 21st Conference Information Technologies - Applications
                  and Theory {(ITAT} 2021), Hotel He{\v{l}}pa, N{\'{\i}}zke Tatry
                  and Mur{\'{a}}nska planina, Slovakia, September 24-28, 2021},
  series       = {{CEUR} Workshop Proceedings},
  volume       = {\textbf{2962}},
  pages        = {140--146},
  publisher    = {CEUR-WS.org},
  year         = {2021},
}

@book{Ra25,
	author		= {Rauch, C.},
	title		= {On the Complexity on the Pumping},
	year		= {2025},
	school   	= {University of Gie{\ss}en, Germany},
	publisher 	= {Logos Berlin}
}

@inproceedings{HoRa25,
  author       = {Holzer, M. and
                  Rauch, C.},
  editor       = {Castiglione, G. and
                  Mantaci, S.},
  title        = {More on Language Families with a Decidable Pumping-Problem (Extended
                  Abstract)},
  booktitle    = {{CIAA} 2025},
  series       = {LNCS},
  volume       = {15981},
  pages        = {223--236},
  publisher    = {Springer},
  year         = {2025}
}

@article{EhPaRo81,
	author       = {A. Ehrenfeucht and
	R. Parikh and
	G. Rozenberg},
	title        = {Pumping Lemmas for Regular Sets},
	journal      = {{SIAM} J. Comput.},
	volume       = {10},
	number       = {3},
	pages        = {536--541},
	year         = {1981}
}

@Inbook{deLuVa99,
	author="de Luca, A.
	and Varricchio, S.",
	title="Regularity Conditions",
	bookTitle="Finiteness and Regularity in Semigroups and Formal Languages",
	year="1999",
	publisher="Springer Berlin Heidelberg",
	address="Berlin, Heidelberg",
	pages="179--194"
}

@Inbook{GrHo25,
	author="Gruber, H. and Holzer, M.",
	editor="Jim{\'e}nez L{\'o}pez and M. Dolores and Vaszil, Gy{\"o}rgy",
	title="On Pumping Constants and Smallest Grammars for Context-Free Languages",
	bookTitle="Languages of Cooperation and Communication: Essays Dedicated to Erzs{\'e}bet Csuhaj-Varj{\'u} to Celebrate Her Scientific Career",
	year="2025",
	publisher="Springer Nature Switzerland",
	address="Cham",
	pages="54--68",
	isbn="978-3-031-97274-4"
}

@misc{MaMaZe26,
      title={The complexity of downward closures of indexed languages}, 
      author={Richard Mandel and Corto Mascle and Georg Zetzsche},
      year={2026},
      eprint={2601.19466},
      archivePrefix={arXiv},
      primaryClass={cs.FL},
      url={https://arxiv.org/abs/2601.19466}, 
}

@inproceedings{DBLP:conf/stacs/Danieli26,
	author       = {Y. Danieli},
	title        = {A Pumping-Like Lemma for Languages over Infinite Alphabets},
	booktitle    = {{STACS} 2026},
	series       = {LIPIcs},
	volume       = {364},
	pages        = {29:1--29:19},
	publisher    = {Schloss Dagstuhl - Leibniz-Zentrum f{\"{u}}r Informatik},
	year         = {2026}
}

@inproceedings{DBLP:conf/stacs/MazowieckiPS26,
	author       = {F. Mazowiecki and
	A. Puch and
	D. Smertnig},
	title        = {Pumping-Like Results for Copyless Cost Register Automata and Polynomially
	Ambiguous Weighted Automata},
	booktitle    = {{STACS} 2026},
	series       = {LIPIcs},
	volume       = {364},
	pages        = {67:1--67:21},
	publisher    = {Schloss Dagstuhl - Leibniz-Zentrum f{\"{u}}r Informatik},
	year         = {2026},
}

@string{jul	= {Juli}}

@string{sep	= {September}}

@string{jul	= {July}}

@string{AI   = {Acta Informatica}}

@string{IBM = {IBM Journal of Research and Development}}

@string{LNCS = {Lecture Notes in Computer Science}}

@string{LIPIcs = {Leibniz International Proceedings in Informatics}}

@string{STACS = {Annual Symposium on Theoretical Aspects of Computing}}

@string{RAIRO = {RAIRO--Informatique th{\'e}orique et Applications / Theoretical Informatics and Applications}}

@string{Springer = {Springer-Verlag}}

@string{EPTCS  = {Electronic Proceedings in Theoretical Computer Science}}

@string{jul	= {}}

@string{sep	= {}}

@string{AI   = {Acta Inform.}}

@string{IBM = {IBM J. Res. Dev.}}

@string{LNCS = {LNCS}}

@string{LIPIcs = {LIPIcs}}

@string{STACS = {Symposium on Theoretical Aspects of Computing}}

@string{Springer = {Springer}}

@string{EPTCS  = {EPTCS}}
	
	\appendix

	
	\printcollectedappendix
	
	
	%
	
\end{document}